\ifx\SOCG\undefined%
\newcommand{\SOCGVer}[1]{}%
\newcommand{\RegVer}[1]{#1}%
\else
\newcommand{\SOCGVer}[1]{#1}%
\newcommand{\RegVer}[1]{}%
\fi

\SOCGVer{
  \documentclass[english]{socg-lipics-v2021}
}
\RegVer{
  \documentclass[english]{lipics-v2021}
  \hideLIPIcs
  \nolinenumbers
}

\usepackage{amsmath,amsthm,amssymb,amsfonts}
\usepackage{graphicx}
\usepackage{algorithm}
\usepackage{hyperref}
\usepackage{url}
\usepackage{color}
\usepackage[dvipsnames]{xcolor}
\usepackage{xspace}
\usepackage{paralist}
\usepackage[normalem]{ulem} 
\usepackage{caption}    
\usepackage{subcaption} 
\usepackage[shortlabels]{enumitem}   
\usepackage{mathtools}

\newtheorem{problem}[theorem]{Problem}
\newtheorem{fact}[theorem]{Fact}

\newcommand{\myparagraph}[1]{\bigskip\noindent{\textbf{#1}}}

\newcommand{\suffixmeb}{\mu}

\newcommand{\remove}[1]{}

\renewcommand{\Re}{\mathbb{R}}
\newcommand{\Term}[1]{\textsf{#1}}
\newcommand{\DAG}{\Term{DAG}\xspace}

\newcommand{\complex}{\mathcal{C}}

\newcommand{\dcA}{\complex_1}
\newcommand{\dcB}{\complex_2}
\newcommand{\Vertices}[1]{V\pth{#1}}
\newcommand{\Edges}[1]{E\pth{#1}}

\newcommand{\curveA}{\pi}

\newcommand{\curveB}{\sigma}

\newcommand{\ball}{B}

\newcommand{\br}[1]{\textcolor{blue}{\textbf{BR}:: {#1}}}

\definecolor{seagreen}{RGB}{45, 138, 86}

\newcommand{\distFr}[2]{\mathsf{d}_{\EuScript{F}}\pth{#1, #2}} 
\newcommand{\distFrW}[2]{\mathsf{d}_{\EuScript{F}}^\mathrm{w} \pth{#1, #2}}
\newcommand{\distDFr}[2]{\mathsf{d}_{\EuScript{DF}}\pth{#1, #2}}
\newcommand{\distDFrW}[2]{\mathsf{d}_{\EuScript{DF}}^\mathrm{w} \pth{#1, #2}}

\newcommand{\edistFr}[2]{\mathsf{ed}_{\EuScript{F}}\pth{#1, #2}}
\newcommand{\edistFrW}[2]{\mathsf{ed}_{\EuScript{F}}^\mathrm{w} \pth{#1, #2}}
\newcommand{\edistDFr}[2]{\mathsf{ed}_{\EuScript{DF}}\pth{#1, #2}}
\newcommand{\edistDFrW}[2]{\mathsf{ed}_{\EuScript{DF}}^\mathrm{w} \pth{#1, #2}}

\newcommand{\normalFontI}{\text{\textnormal{I}}}
\newcommand{\dedistFr}[2]{\mathsf{Ded}_{\EuScript{F}}\pth{#1, #2}}
\newcommand{\iedistFr}[2]{\normalFontI\mathsf{ed}_{\EuScript{F}}\pth{#1, #2}}
\newcommand{\dedistFrW}[2]{\mathsf{Ded}_{\EuScript{F}}^\mathrm{w} \pth{#1, #2}}
\newcommand{\dedistDFr}[2]{\mathsf{Ded}_{\EuScript{DF}}\pth{#1, #2}}
\newcommand{\dedistDFrW}[2]{\mathsf{Ded}_{\EuScript{DF}}^\mathrm{w} \pth{#1, #2}}
\newcommand{\iedistFrW}[2]{\normalFontI\mathsf{ed}_{\EuScript{F}}^\mathrm{w} \pth{#1, #2}}
\newcommand{\iedistDFr}[2]{\normalFontI\mathsf{ed}_{\EuScript{DF}}\pth{#1, #2}}
\newcommand{\iedistDFrW}[2]{\normalFontI\mathsf{ed}_{\EuScript{DF}}^\mathrm{w} \pth{#1, #2}}

\newcommand{\dedistdp}[2]{\mathsf{DedDP} \pth{#1, #2}}
\newcommand{\iedistdp}[2]{\normalFontI\mathsf{edDP} \pth{#1, #2}}
\newcommand{\edistdp}[2]{\mathsf{edDP} \pth{#1, #2}}

\providecommand{\pth}[2][\!]{#1\left({#2}\right)}

\newcommand{\Frechet}{Fr\'{e}chet\xspace}

\newcommand{\pntA}{\mathsf{p}}
\newcommand{\pntB}{\mathsf{q}}

\newcommand{\thresh}{\delta}
\newcommand{\seq}[1]{\langle #1\rangle}

\newcommand{\seclab}[1]{\label{sec:#1}}
\newcommand{\secref}[1]{Section~\ref{sec:#1}}

\newcommand{\apndlab}[1]{\label{apnd:#1}}
\newcommand{\apndref}[1]{Appendix~\ref{apnd:#1}}

\newcommand{\obslab}[1]{\label{observation:#1}}
\newcommand{\obsref}[1]{Observation~\ref{observation:#1}}

\newcommand{\remlab}[1]{\label{remark:#1}}
\newcommand{\remref}[1]{Remark~\ref{remark:#1}}

\newcommand{\figlab}[1]{\label{fig:#1}}
\newcommand{\figref}[1]{Figure~\ref{fig:#1}}

\newcommand{\problab}[1]{\label{problem:#1}}
\newcommand{\probref}[1]{Problem~\ref{problem:#1}}

\newcommand{\factlab}[1]{\label{fact:#1}}
\newcommand{\factref}[1]{Fact~\ref{fact:#1}}

\newcommand{\deflab}[1]{\label{definition:#1}}
\newcommand{\defref}[1]{Definition~\ref{definition:#1}}

\newcommand{\corlab}[1]{\label{corollary:#1}}
\newcommand{\corref}[1]{Corollary~\ref{corollary:#1}}

\newcommand{\lemlab}[1]{\label{lemma:#1}}
\newcommand{\lemref}[1]{Lemma~\ref{lemma:#1}}

\newcommand{\thmlab}[1]{\label{theorem:#1}}
\newcommand{\thmref}[1]{Theorem~\ref{theorem:#1}}

\title{\Frechet Edit Distance}

\author{Emily Fox}{Department of Computer Science, University of Texas at Dallas, USA}{emily.fox@utdallas.edu}{}{Work on this paper was partially supported by NSF CAREER Award 1942597 and CCF Award 2311179.}
\author{Amir Nayyeri}{School of Electrical Engineering and Computer Science,
Oregon State University, USA}{amir.nayyeri@oregonstate.edu}{}{Work on this paper was partially supported by NSF Awards CCF 2311180 and CCF 1941086.}
\author{Jonathan James Perry}{Department of Computer Science, University of Texas at Dallas, USA}{jperry@utdallas.edu}{https://orcid.org/0009-0003-0042-249X}{Work on this paper was partially supported by NSF CAREER
      Award 1750780 and CCF Award 2311179.}
\author{Benjamin Raichel}{Department of Computer Science, University of Texas at Dallas, USA}{benjamin.raichel@utdallas.edu}{https://orcid.org/0000-0001-6584-4843}{Work on this paper was partially supported by NSF CAREER
      Award 1750780 and CCF Award 2311179.}

\authorrunning{E. Fox, A. Nayyeri, J. Perry, and B. Raichel} 

\Copyright{Emily Fox, Amir Nayyeri, Jonathan James Perry, Benjamin Raichel} 



\EventEditors{Wolfgang Mulzer and Jeff M. Phillips}
\EventNoEds{2}
\EventLongTitle{40th International Symposium on Computational Geometry (SoCG 2024)}
\EventShortTitle{SoCG 2024}
\EventAcronym{SoCG}
\EventYear{2024}
\EventDate{June 11-14, 2024}
\EventLocation{Athens, Greece}
\EventLogo{socg-logo.pdf}
\SeriesVolume{293}
\ArticleNo{48}

\begin{document}

\maketitle

\begin{abstract}
 We define and investigate the \Frechet edit distance problem.
 Given two polygonal curves $\curveA$ and $\curveB$ and a threshhold value $\thresh>0$, we seek the minimum number of edits to $\curveB$ such that the \Frechet distance between the edited $\curveB$ and $\curveA$ is at most $\thresh$. For the edit operations we consider three cases, namely, deletion of vertices, insertion of vertices, or both. For this basic problem we consider a number of variants. Specifically, we provide polynomial time algorithms for both discrete and continuous \Frechet edit distance variants, as well as hardness results for weak \Frechet edit distance variants.
\end{abstract}

\section{Introduction}
\seclab{intro}

\subsection{Motivation}
We consider the general problem of shape matching between polygonal curves.
In a standard formulation of this problem, one is given two sequences of points embedded in a common ambient space like \(\mathbb{R}^d\) with \(d\) a constant.
Depending on the specific application, these inputs may be interpreted directly as the discrete point sequences they are or as the vertices of continuous curves obtained by interpolating between contiguous sequence points.

The computational geometry community has strongly promoted the use of the %
\RegVer{continuous and discrete \Frechet distances}\SOCGVer{\Frechet distance}
to handle 
determining similarity of the curves and matching corresponding portions.
The \emph{continuous} \Frechet distance is often presented using the walks of a person and their dog along the curves;
both entities move at any positive variable speed from the beginning to the end of their respective curve, and one seeks the smallest length of a leash needed to keep them connected during their walks.
For the \emph{discrete} variant, the person and dog are replaced by a leashed pair of frogs that iteratively hop between contiguous vertices, either individually or at the same time, and
the length of the leash is only considered during the moments between hops.
Some prior works have also considered the \emph{weak} variants of continuous and discrete \Frechet,  
where the entities are allowed to move \emph{backwards} at times to keep their leashes short.
Beyond its theoretical interest, the \Frechet distance has seen use in mapping and map construction~\cite{akpw-mca-15,cdgnw-ammrf-11}, handwriting recognition~\cite{skb-fdbas-07}, and protein alignment~\cite{jxz-psadf-08}.

We naturally consider two curves to be similar if their \Frechet distance does not exceed some predetermined threshold value~\(\thresh\).
This notion of similarity allows for a single choice of \(\thresh\) that can be used regardless of the curves' length, and it is resilient to differing sampling rates (as long as the sequences are sufficiently dense in the case of discrete \Frechet).
Unfortunately, this intuitively satisfying notion of similarity has some severe issues once we start applying it to noisy real world data such as GPS traces from individuals' phones or vehicles.
In particular, nearly all variants of the \Frechet distance are extremely sensitive to outliers.
Adding even a single point to one of the input curves can increase their distance by an arbitrarily high amount if that point lies far away from the other curve, and this issue is present regardless of how many points are present in the curves' input sequences.
Similarly, a sparsely sampled continuous curve can change dramatically if even a single vertex is ignored.

\subsection{The \Frechet Edit Distance}
Multiple modifications of and even alternatives to the \Frechet distance have been proposed to address the issue described above, and
we review the most relevant of these alternatives in Section~\ref{sec:intro-related}.
At the end of the day, though, we want to keep that our final notion of similarity is based on the standard definitions of \Frechet distance as it remains the best tool we have for working with continuous and densely sampled discrete curves.

We take inspiration from the string edit distance (Levenstein distance).
Viewing the input curves' point sequences as a pair of strings, we ask for the minimum number of edits (point deletions and/or insertions) needed to bring one curve within \Frechet distance \(\thresh\) of the other.
Intuitively, the fewer edits needed, the more likely it is that the input curves really do represent two instances of the same or at least very similar trajectories through the ambient space.
Depending on which of the above variants of the \Frechet distance we use and which edit operations we allow, we obtain one of several specific similarity measures between the curves.
However, we refer to any of these combinations under the general term \emph{\Frechet edit distance}.
We give the formal definitions and notation for these measures in Section~\ref{sec:prelim-fed}.

\subsection{Our Results}

We describe polynomial time algorithms and NP-hardness results for nearly every variant of \Frechet edit distance proposed above.
Let \(m\) and \(n\) denote the number of points in the two input sequences, with \(n\) denoting the number of points in the sequence that can be edited.
\begin{enumerate}[itemindent=20pt,
labelwidth=15pt, labelsep=5pt, listparindent=0.7cm,
align=left]
    \item We describe polynomial time algorithms for certain cases of \Frechet edit distance using the \textbf{strong continuous \Frechet distance}.
    When limited to deletions, in any \(\mathbb{R}^d\)
    we can compute the \Frechet edit distance in \(O(mn^3)\) time.
    If only \(k\) deletions are needed, our algorithm can be made to run in \(O(k^2mn)\) time.
    Further, we can also handle the case when we allow deletions on \emph{either} input curve, and the corresponding running times respectively become \(O(m^3n^3)\) or \(O(k^4mn)\).
    %
    In the plane, \(\mathbb{R}^2\), for insertions only we describe algorithms with times \(O(nm^5)\), or $O(nm^3(k^2+m\log^2 m))$ when limiting to $k$ insertions. When we allow both deletions and insertions these times become $O((m+n)^3nm^3)$ and $O(knm^3(k^2+m\log^2 m))$.

    All of our algorithms for strong continuous \Frechet distance include an embedding of the curve(s) being edited into a \emph{DAG complex}~\cite{hr-fdre-14}, a geometrically embedded directed acyclic graph that represents the different routes one can take through a curve \emph{and its optionally edited portions}.
    For deletions, we include every direct vertex-to-vertex segment in the complex.
    Insertions require substantially more care, because it is not clear ahead of time where one should place the new vertices or where the new subcurves they determine will map to. 
    In fact, a newly inserted subcurve may map to a portion on the curve not being edited that starts or ends on the interior of a segment. 
    Despite this challenge, we can argue that one can restrict attention to a bounded set of canonical subcurves, 
    and these subcurves can be computed with the aid of a result from~\cite{ghms-apsmlp-93} who describe how to compute minimum vertex curves lying within small \Frechet distance to another curve, via the computation of so-called minimum link stabbers.

\item For the \textbf{strong discrete \Frechet distance} with edits limited to deletions, we describe an \(O(mn)\) time algorithm for any pair of curves in \(\mathbb{R}^d\).
    This result cannot be improved upon by any polynomial factor without violating a conditional lower bound known for the discrete \Frechet distance itself~\cite{DBLP:journals/corr/Bringmann14,DBLP:journals/jocg/BringmannM16}.
    For insertions (with or without deletions as well), the running time becomes \(O(m^2 + mn)\).
    These algorithms use relatively straightforward dynamic programming recurrences, although we do some non-trivial precomputation to compute a small set of positions in which to insert new points.

    \item We show that the variant with \textbf{weak discrete \Frechet distance} is NP-hard even for curves in \(\mathbb{R}^1\) when attempting to minimize the number of deletions, minimize the number of insertions, and minimize the number of either kind of edit.
    In fact, even determining if \emph{any} number of deletions leads to small weak discrete \Frechet distance is NP-hard.
    These results can be extended to \textbf{weak continuous \Frechet distance} after moving to the plane \(\mathbb{R}^2\).
    All of our hardness results are shown by a reduction from 3SAT using a similar argument to that used in~\cite{blopuv-cfducod-23} for the hardness of finding a minimum weak discrete \Frechet distance realization for uncertain curves in \(\mathbb{R}^1\).
\end{enumerate}


In addition to deletions and insertions, our results can be extended in a straightforward manner to include \emph{substitutions} as a third possible edit operations for the Fr\'echet edit distance.
We defer the details to the future journal version of the paper.

\subsection{Related and Improved Upon Prior Work}
\label{sec:intro-related}

As far as we are aware, we are the first to consider this particular measure of similarity in full generality, although there is past work that comes close.
The most relevant large body of work concerns the \emph{shortcut \Frechet distance} between curves where one asks for the minimum \Frechet distance possible after replacing disjoint subcurves with line segment shortcuts~\cite{df-jdcfds-13,bds-cfdsn-14,cd-ckfd-22,afkks-dsfds-15,fk-adfdt-20}.
For continuous curves, one can either allow the shortcuts to go between any two (interpolated) points on the curve, or restrict the shortcuts to be between vertices of the curve.
This vertex-to-vertex shortcut restriction is similar to the deletion only version of \Frechet edit distance, except deletion of multiple contiguous points counts as a single shortcut operation.
(By default we assume the shortcut problem is defined without a bound on the number of shortcuts allowed, though the bounded version has also been considered before, and prominently so in \cite{cd-ckfd-22}.)

Most relevant to the current work is a known \(O(n^3 \log n)\) time algorithm for deciding if the continuous \Frechet distance with vertex-to-vertex shortcuts on one of two \(n\)-vertex curves is at most a given value \(\thresh\).
This algorithm is restricted to curves in \(\mathbb{R}^2\)~\cite{bds-cfdsn-14}.
A slight modification to our first algorithm improves the running time to \(O(n^3)\) and works for curves in any \(\mathbb{R}^d\).
We note that the equivalent shortcut problem for the discrete \Frechet distance has a known \emph{linear} time solution in the plane~\cite{afkks-dsfds-15}.
(Recall our discrete deletion only edit distance algorithm minimizes the number of point deletions, and thus its running time cannot see a substantial improvement without violating conditional lower bounds~\cite{DBLP:journals/corr/Bringmann14,DBLP:journals/jocg/BringmannM16}.)
Surprisingly, the shortcut problem becomes NP-hard when shortcuts are allowed between any two points on the continuous curve~\cite{bds-cfdsn-14}. Further developing the hardness picture, our Section~\ref{sec:hardness} result for any number of deletions implies even vertex-to-vertex shortcutting is NP-hard if we switch from the strong to the weak \Frechet distance, 
with the interpretation that the curve must be shortcut before the traversal (i.e.\ one cannot shortcut a subset of vertices and then later go back to a vertex in the subset, which is automatically not possible in the strong version).

Buchin and Pl\"{a}tz~\cite{bp-kfd-22} proposed an alternative to the above problem where one seeks the minimum Fr\'{e}chet distance possible between discrete or continuous curves after removing up to~\(k\) vertices on one or both curves.
By wrapping them in a binary search, their algorithms can be used to solve the deletion only strong Fr\'{e}chet distance versions of our problem.
Our algorithms are faster than theirs by at least a~\(\log n\) factor in every case except allowing deletions from two continuous curves where their algorithm uses one fewer factor of~\(k\).

Leaving behind the \Frechet distance allows one to consider other distance measures that are best defined over the discrete input sequences as opposed to their interpolated curves~\cite{DBLP:journals/corr/AgarwalFPY15,fl-aged-22,Gold:2018:DTW:3266298.3230734,Chen2005,CHEN2004,Marteau2009,Sankararaman:2013:MMS:2525314.2525360,Stojmirovic2009GeometricAO,wmdtsk-exrmd-13}.
Of particular note is the so-called \emph{geometric edit distance} where one attempts to edit one sequence to look \emph{exactly} like the other one, assigning smaller costs for substitutions between nearby points~\cite{DBLP:journals/corr/AgarwalFPY15,fl-aged-22,Gold:2018:DTW:3266298.3230734}.
As opposed to the above measures for discrete sequences, our use of \Frechet distance allows us to work with continuous interpolations of the input sequences.
Even when considering the discrete \Frechet distance, we avoid the issue of two nearly identical but offset curves from having a high distance just because they contain a large number of input points.
If an input resembling two such curves results in a high \Frechet edit distance, it must be because there is a large number of \emph{outlier} points that need to be cleaned up before similarity is evident.

\section{Preliminaries}
\seclab{prelim}

Throughout, given points $p,q\in \Re^d$, $||p-q||$ denotes their Euclidean distance.  
Moreover, given two (closed) sets $P,Q\subseteq \Re^d$, $||P-Q|| = \min_{p\in P, q\in Q} ||p-q||$ denotes their distance, where for a single point $x\in \Re^d$ we write $||x-P||=||\{x\}-P||$. 
$\ball(x,r)$ will be used to denote the closed ball of radius $r$ centered at $x$.
We use angled brackets to denote an ordered list $\seq{x_1,\ldots,x_n}$, and use $L_1\circ L_2$ to denote the concatenation of ordered lists $L_1$ and $L_2$, where for a single item $x$ we sometimes write $x\circ L = \seq{x}\circ L$. 

\myparagraph{\Frechet Distance.}
A \emph{polygonal curve} of length $m$ is a sequence of $m$ points $\curveA=\seq{\curveA_1,\ldots,\curveA_m}$ where $\curveA_i\in \Re^d$ for all $i$. 
Such a sequence induces a continuous mapping from $[1,m]$ to  $\mathbb{R}^d$, which we also denote by $\curveA$, such that for any integer $1\leq i< m$, the restriction of $\curveA$ to the interval $[i,i+1]$ is defined by $\curveA(i+\alpha) = (1-\alpha)\curveA_i+\alpha\curveA_{i+1}$ for any $\alpha\in [0,1]$, i.e.\ a straight line segment. 
We will view $\curveA$ as both a discrete point sequence and a continuous function interchangeably, and when it is clear from the context, we also may use $\curveA$ to denote the image $\curveA([1,m])$.  
We use \(\curveA[i, j]\), for \(i \leq j\), to denote the restriction of \(\curveA\) to the interval \([i, j]\). Given a curve $\curveA=\seq{\curveA_1,\ldots,\curveA_m}$, we write $|\curveA|=m$ to denote its size.

A reparameterization for a curve $\curveA$ of length $m$ is a continuous non-decreasing bijection $f:[0,1]\rightarrow [1,m]$ such that $f(0) = 1,f(1)=m$.
Given reparameterizations $f,g$ of an $m$ length curve $\curveA$ and an $n$ length curve $\curveB$, respectively, the \emph{width} between $f$ and $g$ is defined as
\RegVer{\[width_{f,g}(\curveA,\curveB)= \max_{\alpha\in[0,1]} ||\curveA(f(\alpha))-\curveB(g(\alpha)) || .\]}
\SOCGVer{\(width_{f,g}(\curveA,\curveB)= \max_{\alpha\in[0,1]} ||\curveA(f(\alpha))-\curveB(g(\alpha)) ||\).}
The (standard, i.e.\ continuous and strong) \emph{Fr\'{e}chet distance} between $\curveA$ and $\curveB$ is then 
\RegVer{\[\distFr{\curveA}{\curveB}=\inf_{f,g} width_{f,g}(\curveA,\curveB)\]}
\SOCGVer{\(\distFr{\curveA}{\curveB}=\inf_{f,g} width_{f,g}(\curveA,\curveB)\)}
where $f,g$ range over all possible reparameterizations of $\curveA$ and $\curveB$.

\RegVer{In this paper we will consider several standard variants of the \Frechet distance.}
The \emph{discrete \Frechet distance} is similar to the above defined \Frechet distance, except that we do not traverse the edges but rather discontinuously jump to adjacent vertices. Specifically, define a monotone correspondence as a sequence of index pairs $\seq{(i_1,j_1),\ldots,(i_k,j_k)}$ such that $(i_1,j_1)=(1,1)$, $(i_k,j_k)=(m,n)$, for any $1\leq z\leq k$ we have $1\leq i_z\leq m$ and $1\leq j_z\leq n$, and for any $1\leq z<k$ we have $(i_{z+1},j_{z+1})\in \{(i_{z}+1,j_{z}),(i_{z},j_{z}+1),(i_{z}+1,j_{z}+1)\}$. Let $C$ denote the set of all monotone correspondences, then the discrete \Frechet distance is 
\RegVer{\[
\distDFr{\curveA}{\curveB} = \inf_{c\in C} \max_{(i,j)\in c} ||\curveA_{i}-\curveB_{j}||
.\]}
\SOCGVer{\(\distDFr{\curveA}{\curveB} = \inf_{c\in C} \max_{(i,j)\in c} ||\curveA_{i}-\curveB_{j}||\).}

Both the \Frechet distance and the discrete \Frechet distance have a corresponding \emph{weak} variant, which is defined analogously except that one is allowed to backtrack on the curves. Specifically, the \emph{weak \Frechet distance}, denoted $\distFrW{\curveA}{\curveB}$, is defined similarly to the standard \Frechet distance above, except that when defining the width $f$ and $g$ are no longer required to be non-decreasing bijections, but are still required to be continuous and have $f(0)=1, g(0)=1$ and $f(1)=m,g(1)=n$. Similarly, the \emph{weak discrete \Frechet distance}, denoted $\distDFrW{\curveA}{\curveB}$, is defined similarly to the discrete \Frechet distance above, except that we no longer require the correspondence to be monotone. Specifically, a (non-monotone) correspondence is sequence of index pairs $\seq{(i_1,j_1),\ldots,(i_k,j_k)}$ such that $(i_1,j_1)=(1,1)$, $(i_k,j_k)=(m,n)$, for any $1\leq z\leq k$ we have $1\leq i_z\leq m$ and $1\leq j_z\leq n$, and for any $1\leq z<k$ we have $(i_{z+1},j_{z+1})\in \{(i_{z}\pm 1,j_{z}),(i_{z},j_{z} \pm 1),(i_{z}\pm 1,j_{z}\pm 1)\}$.

\myparagraph{Free Space.}
To compute the \Frechet distance one normally looks at the so called \emph{free space}.
For the continuous case, the $\thresh$ free space between curves $\curveA$ and $\curveB$, of sizes $m$ and $n$ respectively, is defined as
\[F_\thresh=\{(\alpha,\beta)\in [1,m]\times [1,n]\ \mid\ ||\curveA(\alpha)-\curveB(\beta)|| \leq \thresh\}.\]
%


Alt and Godau \cite{ag-cfdpc-95} observed that any $x,y$ monotone path in the $\thresh$ free space from $(1,1)$ to $(m,n)$ corresponds to a pair of reparameterizations $f$, $g$ of $\curveA$, $\curveB$ such that $width_{f,g}(\curveA,\curveB)\leq \thresh$. The converse also holds and hence $d_{\cal F}(\curveA,\curveB)\leq \thresh$ if and only if such a monotone path exists. Thus in order to determine if $d_{\cal F}(\curveA,\curveB)\leq \thresh$, we define the reachable free space,
\[R_{\thresh}=\{(\alpha,\beta)\in F_\thresh \ \mid\ \text{there exists an $x,y$ monotone path from $(1,1)$ to $(\alpha,\beta)$}\}.\]
Hence  $d_{\cal F}(\curveA,\curveB)\leq \thresh$ if and only if $(m,n)\in R_{\thresh}$.

\newcommand{\freespacedetails}{
$C(i,j) = [i,i+1]\times [j,j+1]$ is referred to as the \emph{cell} of the free space diagram determined by edges $\curveA_{i}\curveA_{i+1}$ and $\curveB_{j}\curveB_{j+1}$. Alt and Godau \cite{ag-cfdpc-95} made the crucial observation that the free space within any cell, i.e.\ $F_\thresh(i,j)=F_\thresh\cap C(i,j)$ is always a convex set (specifically, the clipping of an affine transformation of a disk to the cell). Thus one can restrict attention solely to the free and reachable spaces restricted to the boundaries of each cell. 
Specifically, let $L^F_{i,j}$ (resp.\ $B^F_{i,j}$) denote the left (resp.\ bottom) free space interval of $C(i,j)$, i.e.\ $L^F_{i,j} = F_\thresh(i,j) \cap (\{i\}\times [j,j+1])$ (resp. $B^F_{i,j} = F_\thresh(i,j) \cap ([i,i+1]\times \{j\})$). Analogously define $L^{R}_{i,j}$ and $B^{R}_{i,j}$ for the reachable subsets of the left and bottom boundaries of $C(i,j)$. By convexity of the free space, given $L^{R}_{i,j}$ and $B^{R}_{i,j}$, one can determine $L^{R}_{i+1,j}$ by setting $L^{R}_{i+1,j}=L^{F}_{i+1,j}$ if $B^{R}_{i,j}$ is non-empty, and otherwise $L^{R}_{i+1,j}$ is the subinterval of $L^{F}_{i+1,j}$ whose $y$-coordinate is greater than or equal to the smallest $y$-coordinate of a point in $L^{R}_{i,j}$. $B^{R}_{i,j+1}$ is computed analogously. 
Finally, since the cells of the free space have a natural partial order (i.e.\ $C(i,j)\preceq C(k,l)$ if and only if $i\leq k$ and $j\leq l$), the reachable intervals can now be propagated cell by cell according to a topological sorting of the cells. This determines whether $(m,n)\in R_{\thresh}$, and hence whether $d_{\cal F}(\curveA,\curveB)\leq \thresh$, in linear time in the number of cells in the free space (i.e.\ $O(mn)$ time).
}
\RegVer{\freespacedetails}

For the case of the discrete \Frechet distance, the free space can still be considered, and is simply described by an $m\times n$ grid graph. Specifically, the vertices are all pairs $(i,j)$ such that $1\leq i\leq m$ and $1\leq j\leq n$, and for any vertex $(i,j)$ we create the directed edges $(i,j)\rightarrow (i+1,j)$, $(i,j)\rightarrow (i,j+1)$, and $(i,j)\rightarrow (i+1,j+1)$ (whenever the corresponding destination vertex exists). A vertex $(i,j)$ is then called \emph{free} if $||\curveA_i-\curveB_j||\leq \delta$. Analogous to the continuous case, we then have that $\distDFr{\curveA}{\curveB}\leq \delta$ if and only if there is a path in this directed graph from $(1,1)$ to $(m,n)$ which only uses free vertices.

For the weak discrete \Frechet distance the free space is described by the undirected graph on the same set of vertices, where vertex $(i,j)$ and vertex $(i',j')$ are adjacent if and only if $|i-i'|\leq 1$ and $|j-j'|\leq 1$.
Again, $\distDFrW{\curveA}{\curveB}\leq \delta$ if and only if there is a path in this undirected graph from $(1,1)$ to $(m,n)$ which only uses free vertices. 
Analogously, the free space for the weak continuous \Frechet distance is the same as that for the strong continuous \Frechet distance, but now we no longer require the path through the free space be monotone.  

\myparagraph{\Frechet Edit Distance.}
\seclab{prelim-fed}
Given a curve $\curveB=\seq{\curveB_1,\ldots,\curveB_n}$, a \emph{deletion} of the vertex $\curveB_i$ produces the $n-1$ vertex curve $\curveB' = \seq{\curveB_1,\ldots,\curveB_{i-1},\curveB_{i+1},\ldots,\curveB_n}$.  Conversely the \emph{insertion} of a vertex $\pntA$ into $\curveB$ at position $i$ produces the $n+1$ vertex curve $\curveB' = \seq{\curveB_1,\ldots,\curveB_{i-1},\pntA, \curveB_{i},\ldots, \curveB_n}$. Both deletions and insertions are referred to as \emph{edits}.
\RegVer{

}
Let $\thresh>0$ be a fixed threshold distance. Then given polygonal curves $\curveA=\seq{\curveA_1,\ldots,\curveA_m}$ and $\curveB=\seq{\curveB_1,\ldots,\curveB_n}$ define the $\delta$-threshold \emph{\Frechet edit distance} from $\curveB$ to $\curveA$ as the minimum number of edits to $\curveB$, producing a new curve $\curveB'$, such that $\distFr{\curveA}{\curveB'}\leq \delta$. As $\delta>0$ is some fixed value, and the term ``\Frechet'' is implicit, throughout we refer to this more simply as the \emph{edit distance} from $\curveB$ to $\curveA$, and we denote it as $\edistFr{\curveA}{\curveB}$. 
We analogously define the weak edit distance, denoted $\edistFrW{\curveA}{\curveB}$, the discrete edit distance, denoted $\edistDFr{\curveA}{\curveB}$, and the weak discrete edit distance, denoted $\edistDFrW{\curveA}{\curveB}$, by replacing the condition $\distFr{\curveA}{\curveB'}\leq \delta$ with $\distFrW{\curveA}{\curveB'}\leq \delta$, $\distDFr{\curveA}{\curveB'}\leq \delta$, and $\distDFrW{\curveA}{\curveB'}\leq \delta$, respectively.

For any one of these variants we may consider the case when only deletions or only insertions are allowed. In this case we prepend $\mathsf{D}$ for deletion only, or $\normalFontI$ for insertion only. (For example, $\edistFr{\curveA}{\curveB}$ becomes $\dedistFr{\curveA}{\curveB}$ or $\iedistFr{\curveA}{\curveB}$.) Note that by considering only deletions or only insertions, there may be no valid edit sequence, in which case we define the edit distance as $\infty$.  Conversely, if we allow both deletions and insertions, there is always a solution of cost $m+n$ by deleting all vertices of $\curveB$ and inserting all vertices of $\curveA$.

\section{\DAG Complexes}\seclab{ogdagcomplex}

\cite{hr-fdre-14} define the following generalization of a curve. 
Consider a directed acyclic graph (\DAG)  with vertices in $\Re^d$, where a directed edge $\pntA\rightarrow\pntB$ is realized by the directed segment $\pntA\pntB$. We refer to such an embedded graph as being a \emph{\DAG complex}, denoted $\complex$, with embedded vertices $\Vertices{\complex}$ (i.e.\ points) and embedded edges $\Edges{\complex}$ (i.e.\ line segments). We assume the underlying graph is weakly connected and thus write $|\complex| = |\Edges{\complex}|$. Note also that a \DAG complex is allowed to have crossing edges or overlapping vertices (i.e.\ it is not necessarily an embedding in $\Re^d$).
Call a polygonal curve $\curveA=\seq{\curveA_1,\ldots,\curveA_k}$ \emph{compliant} with $\complex$ if $\curveA_i\in V(\complex)$ for all $i$ and $\curveA_i\curveA_{i+1}\in E(\complex)$ for all $1\leq i<k$. (Note this implies $\curveA$ traverses each edge in the direction compliant with its orientation from the \DAG.)
%
\cite{hr-fdre-14} considered the following.

\begin{problem}\problab{dag}
Given two \DAG complexes $\dcA$ and $\dcB$,
start vertices $s_1 \in \Vertices{\dcA},s_2 \in
\Vertices{\dcB}$, end vertices $t_1 \in
\Vertices{\dcA}, t_2 \in \Vertices{\dcB}$, and a value $\delta$, determine if there exists two polygonal curves $\curveA_1, \curveA_2$, such that:
    \RegVer{\begin{enumerate}[(a), leftmargin=1.5cm]}
    \SOCGVer{\begin{inparaenum}[(a)]}
    \item $\curveA_i$ is compliant with $\complex_i$ for $i=1,2$.
    \item $\curveA_i$ starts at $s_i$ and ends at $t_i$ in
    $\complex_i$, for $i=1,2$.
    \item $\distFr{\curveA_1}{\curveA_2} \leq \delta$.
    \RegVer{\end{enumerate}}
    \SOCGVer{\end{inparaenum}}
\end{problem}

\SOCGVer{
\cite{hr-fdre-14} solve \probref{dag} in $O(|\complex_1||\complex_2|)$ time by considering the free space of the product complex of $\dcA$ and $\dcB$.
This is analogous to the standard procedure used for the \Frechet distance between curves.
In the full version, we describe this standard procedure for curves and then how the proceedure extends to this more general product complex. This in turn allows us to remark how the procedure from \cite{hr-fdre-14} can easily be extended to the more general setting where we allow multiple starting and ending points, resulting in the following theorem.

}

\newcommand{\dagcomplex}{
\cite{hr-fdre-14} solve \probref{dag} in $O(|\complex_1||\complex_2|)$ time by considering the free space of the product complex of $\dcA$ and $\dcB$.
This is analogues to the procedure described above  for the standard \Frechet distance between curves, which are a special case of \DAG complexes.

We now describe the product complex in more detail (in a manor slightly more tailored to our purposes than in \cite{hr-fdre-14}), which will allow us to remark how the procedure from \cite{hr-fdre-14} can easily be extended to the more general setting where we allow multiple starting and ending points. 

Given two \DAG complexes $\complex_1$ and $\complex_2$, their product complex $\complex = \complex_1\times \complex_2$ consists of a collection of cells, and a description of which cells are adjacent along different boundary edges. Specifically, each cell is the Cartesian product of a pair of (ordered) segments $uu'\in E(\complex_1)$ and $vv'\in E(\complex_2)$. Note that as $uu'$ and $vv'$ are ordered, their resulting cell has a well defined left, right, bottom, and top boundary edge. 
For two cells $uu'\times vv'$ and $ww'\times vv'$, if $w=u'$ then we identify the right bounding edge of the cell $uu'\times vv'$ with the left bounding edge of the cell $ww'\times vv'$, i.e.\ the cells are adjacent along this edge. 
Similarly, for cells $uu'\times vv'$ and $uu'\times ww'$, if $w=v'$ then we identify the bottom edge of $uu'\times ww'$ with the top edge of $uu'\times vv'$.
(Note we view $\complex$ as an abstract complex, and are not concerned with whether it can be embedded such that non-adjacent cells do not intersect.)

The product complex of two \DAG complexes is thus analogous to the standard free space diagram for curves described above, except that multiple cells can be adjacent along a given cell boundary edge. In particular, by the same reasoning the free space within each cell of the product complex is convex, and thus we only need to propagate reachability information on the bounding edges of the cells. Finally, the topological orderings of the \DAG complexes $\complex_1$ and $\complex_2$ induce a topological ordering of the cells of $\complex$, and thus there is a valid ordering to propagate the reachability information. Specifically, for a cell $uu'\times vv'$, the reachability interval of its left bounding edge $u\times vv'$ depends on all adjacent cells whose right boundary edge is $u\times vv'$. Any such adjacent cell precedes the cell $uu'\times vv'$ in the topological ordering. Thus we have the reachability intervals on the left and bottom boundaries of this adjacent cell, and so can propagate reachability to the edge $u\times vv'$ in an identical manner to that described above for the standard free space diagram between curves. Doing this for all adjacent cells produces a collection of reachability intervals on $u\times vv'$, and we simply take the union of all such intervals. Observe that from the description above for curves, any non-empty reachability interval on $u\times vv'$ always has top endpoint equal to the top endpoint of the free space on this edge, and therefore the union of these intervals is itself a single interval (and thus computable in linear time in the number of adjacent cells). 

In \cite{hr-fdre-14} there was a single starting point in each \DAG complex, $s_1\in V(\complex_1)$ and $s_2\in V(\complex_2)$. This determines a single vertex $(s_1, s_2)$. If this vertex is not in the free space (i.e.\ $||s_1-s_2||> \thresh$) then there are no reachable points. Otherwise, if $(s_1, s_2)$ is free, then the reachability is propagated from this reachable starting vertex in topological order as described above. (In particular the entire free space on any edge of the form $s_1\times s_2v$ or $s_1v\times s_2$, for any $v$, is initially marked a reachable.) 
We generalize this to allow sets of starting vertices $S_1\subseteq V(\complex_1)$ and $S_2\subseteq V(\complex_2)$. Now we simply apply the same reachable initialization process for any pair $s_1\in S_1$ and $s_2\in S_2$, i.e.\ if $(s_1, s_2)$ is free then we mark it as reachable (and the entire free space on edges of the form $s_1\times s_2v$ or $s_1v\times s_2$). 
Now propagate the reachable intervals according to the topological ordering of the cells. Note that in the above propagation procedure, for a given cell boundary edge we already were taking the union of reachable intervals propagated from adjacent cells, so the only difference is that now, if the edge was initially marked as reachable, then this union potentially contains one more interval.\footnote{In fact, in this case the union is the entire free space interval on the edge.}
Similarly, \cite{hr-fdre-14} considered a single target vertex in each \DAG complex, $t_1\in V(\complex_1)$ and $t_2\in V(\complex_2)$, whereas we will consider sets of target vertices, $T_1\subseteq V(\complex_1)$ and $T_2\subseteq V(\complex_2)$, and we wish to determine all reachable pairs $(t_1, t_2)$ such that $t_1\in T_1$ and $t_2\in T_2$. Note that as described above, we are already propagating the reachability information to the entire product complex, and thus this generalization to sets of target vertices does not require any modification to the algorithm. }
\RegVer{%
\dagcomplex
Thus in summary we have the following theorem.
}

\begin{theorem}\thmlab{dag}
Given two \DAG complexes $\complex_1$ and $\complex_2$, starting vertices $S_1\subseteq V(\complex_1)$ and $S_2\subseteq V(\complex_2)$, target vertices $T_1\subseteq V(\complex_1)$ and $T_2\subseteq V(\complex_2)$, and a value $\delta$, then in $O(|\complex_1||\complex_2|)$ time one can determine the set of all pairs $t_1\in T_1$ and $t_2\in T_2$, such that there are curves $\curveA_1$ and $\curveA_2$ such that 
    \RegVer{\begin{enumerate}[(a), leftmargin=1.5cm]}
    \SOCGVer{\begin{inparaenum}[(a)]}
    \item $\curveA_i$ is compliant with $\complex_i$ for $i=1,2$.
    \item $\curveA_i$ starts at some $s_i\in S_i$ and ends at $t_i$, for $i=1,2$.
    \item $\distFr{\curveA_1}{\curveA_2} \leq \delta$.
    \RegVer{\end{enumerate}}
    \SOCGVer{\end{inparaenum}}
\end{theorem}





\section{Continuous \Frechet Distance}

We give algorithms to compute $\dedistFr{\curveA}{\curveB}$,  $\iedistFr{\curveA}{\curveB}$, and $\edistFr{\curveA}{\curveB}$.
The high level approach in each case is to convert $\curveA$ and $\curveB$ into \DAG complexes and apply \thmref{dag}.

Recall that in the \Frechet edit distance problems, we are only editing $\curveB$, not $\curveA$. As remarked above, $\curveA$ is itself a \DAG complex, and using this complex directly represents that $\curveA$ is not modified. Thus in the following the task is to model edits to $\curveB$ with an appropriate \DAG complex. (For $\dedistFr{\curveA}{\curveB}$ we will remark that creating such \DAG complexes for both $\curveA$ and $\curveB$ allows modelling the problem where deletion is allowed on either curve.) 

\subsection{Deletion Only}\seclab{cdo}
Given a curve $\curveB=\seq{\curveB_1,\ldots,\curveB_n}$, consider the \DAG complex produced by adding all possible forward edges to $\curveB$, namely all directed edges $\curveB_i\curveB_j$ for all $1\leq i<j\leq n$. We will refer to this as the \emph{complete \DAG complex} induced by $\curveB$.
Observe that any curve that is compliant with the complete \DAG complex is defined precisely by the subsequence of vertices from $\curveB$ it contains. Thus the set of curves that are compliant with the complete \DAG complex is in one to one correspondence with the set of subsequences of $\curveB$. Conversely, any curve obtained by deleting a subset of vertices from $\curveB$, is defined by the subsequence of $\curveB$ that remains. Thus one concludes that the set of all curves that are compliant with the complete \DAG complex of $\curveB$ are in one to one correspondence with the set of curves obtainable from $\curveB$ by deletions.

The above tells us that the complete \DAG complex encodes all possible curves produced by deletion, however, it needs to be further modified to also encode the cost of these deletions. To account for this cost we make $k$ additional copies of $\curveB$, where $k$ is some bound on the number of allowed deletions (which may be as large as $n$). Intuitively, the copy number of a given vertex encodes the number of deletions made so far.
So let $\curveB^\ell=\seq{\curveB_1^\ell,\ldots,\curveB_n^\ell}$ denote the $\ell$th copy.
Then to construct the \DAG complex, for all $0\leq \ell\leq k$ and all $i<j$ such that $\ell+(j-(i+1)) \leq k$, we add the directed edge $\curveB_i^{\ell}\curveB_j^{\ell+(j-(i+1))}$. Such edges are added since if we wish to delete all vertices between $\curveB_i$ and $\curveB_j$ (and hence use the edge $\curveB_i\curveB_j$) then we pay for these $(j-(i+1))$ deletions by advancing from the copy $\ell$ to copy $\ell+(j-(i+1))$ of $\curveB$. Call the resulting complex the \emph{complete weighted \DAG complex} of $\curveB$.

Now given $\curveA=\seq{\curveA_1,\ldots,\curveA_m}$ and $\curveB=\seq{\curveB_1,\ldots,\curveB_n}$, our goal is to decide if $\dedistFr{\curveA}{\curveB}\leq k$. As discussed above, the directed edges of $\curveA$ immediately define a \DAG complex, and thus we refer to this complex simply as $\curveA$.
On the other hand, for $\curveB$ we construct the complete weighted \DAG complex for $\curveB$, denoted $\complex_\curveB$.
Now for $\curveA$ we must start at $\curveA_1$ and end at $\curveA_m$, however, for $\curveB$ the optimal solution may delete some prefix of vertices $\curveB_1,\ldots,\curveB_{i}$, which would correspond to starting at vertex $\curveB_{i+1}^{i}$ in $\complex_\curveB$. Thus the set $S_\curveB$ of starting vertices consists of all vertices $\curveB_{i+1}^{i}$. Similarly, the optimal solution may delete some suffix of vertices from $\curveB$. To handle this case, however, we simply consider all possible ending vertices, namely $T_\curveB=V(\complex_\curveB)$.   
Then we call \thmref{dag}, which in $O(k^2mn)$ time (since $|\curveA|=O(m)$ and $|\complex_\curveB|=O(k^2n)$) computes the set of all pairs in $\curveA_m\times V(\complex_\curveB)$ such that there are compliant paths from allowable starting vertices whose \Frechet distance is $\leq \delta$. If no such pair exists then $\dedistFr{\curveA}{\curveB}> k$. Otherwise, let $(\curveA_m,\curveB_i^{\alpha})$ be one of the computed ending pairs. Then reaching this pair corresponds to deleting $\alpha$ vertices before $\curveB_i$, plus deleting all $n-i$ vertices after $\curveB_i$. Thus for each such pair $(\curveA_m,\curveB_i^{\alpha})$ we check if $\alpha+(n-i)\leq k$, and if this holds for some pair then $\dedistFr{\curveA}{\curveB}\leq k$, and otherwise $\dedistFr{\curveA}{\curveB}> k$.

Before stating our summarizing theorem, we observe several easy extensions. 
First, if deletions are allowed on both curves, then the same procedure works where instead of using $\curveA$ as one of the \DAG complexes, we use the complete weighted \DAG complex $\complex_\curveA$, yielding $O(k^4mn)$ time in total. 
Alternatively, again only allow deletions on $\curveB$, but consider the problem of  computing $\dedistFr{\curveA}{\curveB}$, rather than determine if $\dedistFr{\curveA}{\curveB}\leq k$ for some $k$. In this case, the same procedure works by setting $k=n$ (as one cannot delete more vertices than the curve contains), and then finding the pair $(\curveA_m,\curveB_i^{\alpha})$ of allowable end vertices minimizing $\alpha+(n-i)$, resulting in an $O(mn^3)$ running time. 
Finally, applying this same idea to computing $\dedistFr{\curveA}{\curveB}$ when deletions are allowed on both curves gives an $O(m^3n^3)$ time, as there can be at most $n$ deletions on $\curveB$ and at most $m$ on $\curveA$.

\begin{theorem}
Given curves $\curveA=\seq{\curveA_1,\ldots,\curveA_m}$ and $\curveB=\seq{\curveB_1,\ldots,\curveB_n}$, a threshold $\thresh$, and an integer parameter $k>0$, in $O(k^2mn)$ time one can determine if $\dedistFr{\curveA}{\curveB}\leq k$.

If deletions are allowed on both $\curveA$ and $\curveB$, then in $O(k^4mn)$ time one can determine if $\dedistFr{\curveA}{\curveB}\leq k$. 
Finally, one can compute $\dedistFr{\curveA}{\curveB}$ in $O(mn^3)$ time if deletions are only allowed on $\curveB$, and in $O(m^3n^3)$ time if deletions are allowed on both curves.
\end{theorem}

\SOCGVer{A slight variation of the above algorithm can be used to solve the vertex restricted shortcut \Frechet distance problem as described in \cite{df-jdcfds-13,bds-cfdsn-14}, improving the result in \cite{bds-cfdsn-14} for \(\Re^2\) by a \(\log n\) factor while also extending it to \(\Re^d\).
Details appear in the full version.}
\newcommand{\shortcutdetails}{
We now describe how the algorithm for continuous strong \Frechet distance with deletions only can be modified to handle the vertex restricted shortcut problem considered in \cite{df-jdcfds-13, bds-cfdsn-14}.  In this problem, $\curveA$ is fixed, and on $\curveB$ you are allowed to shortcut (i.e.\ take the line segment) directly from $\curveB_i$ to $\curveB_j$ for any $i<j$. Here you are allowed to shortcut as often as you like and for free and the question is whether you can get the \Frechet distance between the resulting curves to be $\leq \thresh$. This effectively is the zero cost version of our deletion only problem, although deleting $\curveB_1$ and $\curveB_n$ is not allowed. As described above the zero cost case is modeled by the (unweighted) complete \DAG complex induced by $\curveB$, i.e.\ we add all forward edges but we do not need to create multiple copies of $\curveB$. Moreover, to model that deleting $\curveB_1$ and $\curveB_n$ is not allowed we simply set our starting set $S_\curveB=\{\curveB_1\}$ and our ending set $T_\curveB=\{\curveB_n\}$.
}%
\RegVer{\shortcutdetails
We thus have the following corollary, which is faster than the result in \cite{bds-cfdsn-14} by a $\log n$ factor, and works for curves in $\Re^d$, whereas the result in \cite{bds-cfdsn-14} is restricted to curves in $\Re^2$.
}

\begin{corollary}
\label{cor:shortcuts}
Given a threshold $\thresh$, a fixed curve $\curveA=\seq{\curveA_1,\ldots,\curveA_m}$, and  a curve $\curveB=\seq{\curveB_1,\ldots,\curveB_n}$ which allows shortcuts, then in $O(mn^2)$ time one can determine if the vertex restricted shortcut \Frechet distance is $\leq \delta$. 
\end{corollary}

\newcommand{\mv}[1]{\mathsf{mv}_\thresh(#1)}
\newcommand{\enter}[1]{\mathsf{enter}_\thresh(#1)}
\newcommand{\leave}[1]{\mathsf{leave}_\thresh(#1)}
\newcommand{\clip}[1]{\mathsf{clip}_\thresh(#1)}
\newcommand{\canon}{\mathsf{CS}}
\newcommand{\ecanon}{\mathsf{ECS}}
\RegVer{
\subsection{Insertion Only}}
\SOCGVer{\subsection{Insertions}}
\seclab{cio}

\SOCGVer{
Applying the above approach for insertions only or both deletions and insertions is considerably more difficult.
We sketch the main argument here, and give the full details in the full version.
For this section, we assume both \(\curveA\) and \(\curveB\) are in \(\Re^2\), since we use the results in $\Re^2$ from \cite{ghms-apsmlp-93}.}
\RegVer{
Given curves $\curveA=\seq{\curveA_1,\ldots,\curveA_m}$ and $\curveB=\seq{\curveB_1,\ldots,\curveB_n}$ and a  threshold $\thresh$, our goal in this section is to compute $\iedistFr{\curveA}{\curveB}$. Note that insertions are only allowed on $\curveB$, and in this section we will now  assume that both $\curveA$ and $\curveB$ are in $\Re^2$.}
For simplicity, we will first assume there are no insertions before $\curveB_1$ nor after $\curveB_n$.

Observe that if it is beneficial to insert a subcurve between two consecutive vertices of $\curveB$, then this subcurve should be a minimum vertex curve with \Frechet distance $\thresh$ to some portion of $\curveA$. 
\RegVer{Thus before giving our algorithm, we give background regarding such minimum vertex curves.}
Unfortunately, the portion of $\curveA$ that we are matching to may not begin and end on vertices of $\curveA$.
Regardless, it suffices to consider a bounded number of canonical starting and ending location pairs.

\SOCGVer{
\begin{definition}\deflab{mv2}
Given a curve $\curveA=\seq{\curveA_1,\ldots,\curveA_m}$, a value $\thresh$, 
and points $s$ and $t$ such that $||s-\curveA_1||\leq \thresh$ and $||t-\curveA_m||\leq \thresh$, 
let $\mv{s,t,\curveA}$ denote the curve $\curveB=\seq{\curveB_1,\ldots,\curveB_n}$ with the minimum number of vertices such that 
$\distFr{\curveA}{s\circ\curveB\circ t}\leq \delta$.

For an ordered segment $q_1q_2$ and a point p such that $B(p,\thresh)\cap q_1q_2\neq \emptyset$, let $\enter{p,q_1q_2}$ denote the point in $B(p,\thresh)\cap q_1q_2$ closest to $q_1$, and similarly let $\leave{p,q_1q_2}$ denote the point in $B(p,\thresh)\cap q_1q_2$ closest to $q_2$.
Finally, given a curve $\curveA=\seq{\curveA_1,\ldots,\curveA_m}$ where $m>2$, and points $s$ and $t$ such that $||s-\curveA_1\curveA_2||\leq \thresh$ and $||t-\curveA_{m-1}\curveA_m||\leq \thresh$, define $\clip{s,t,\curveA}=\seq{\leave{s,\curveA_1\curveA_2},\curveA_2,\ldots,\curveA_{m-1},\enter{t,\curveA_{m-1}\curveA_m}}$
\end{definition}
Let $\mv{\curveA}$ be the analogue of $\mv{s,t,\curveA}$ from \defref{mv2}, except where we require $\distFr{\curveA}{\curveB}\leq \delta$ instead of $\distFr{\curveA}{s\circ\curveB\circ t}\leq \delta$, i.e.\ the starting points $s$ and $t$ are not specified. \cite{ghms-apsmlp-93} compute $\mv{\curveA}$ in their Theorem 14. The full version reduces $\mv{s,t,\curveA}$ to $\mv{\curveA}$.

\begin{theorem}[\cite{ghms-apsmlp-93}]
Given $\curveA=\seq{\curveA_1,\ldots,\curveA_m}$, a value $\thresh$, 
and points $s$ and $t$ such that $||s-\curveA_1||\leq \thresh$ and $||t-\curveA_m||\leq \thresh$,
then $\mv{s,t,\curveA}$ can be computed in $O(m^2\log^2 m)$ time.
\end{theorem}
}

\newcommand{\continuousinsert}{
\subsubsection{Minimum Vertex Curves}

\begin{definition}\deflab{mv}
Given a curve $\curveA=\seq{\curveA_1,\ldots,\curveA_m}$, a value $\thresh$, 
and points $s$ and $t$ such that $||s-\curveA_1||\leq \thresh$ and $||t-\curveA_m||\leq \thresh$, 
let $\mv{s,t,\curveA}$ denote the curve $\curveB=\seq{\curveB_1,\ldots,\curveB_n}$ with the minimum number of vertices such that 
$\distFr{\curveA}{s\circ\curveB\circ t}\leq \delta$.

For an ordered segment $q_1q_2$ and a point p such that $B(p,\thresh)\cap q_1q_2\neq \emptyset$, let $\enter{p,q_1q_2}$ denote the point in $B(p,\thresh)\cap q_1q_2$ closest to $q_1$, and similarly let $\leave{p,q_1q_2}$ denote the point in $B(p,\thresh)\cap q_1q_2$ closest to $q_2$.
Finally, given a curve $\curveA=\seq{\curveA_1,\ldots,\curveA_m}$ where $m>2$, and points $s$ and $t$ such that $||s-\curveA_1\curveA_2||\leq \thresh$ and $||t-\curveA_{m-1}\curveA_m||\leq \thresh$, define $\clip{s,t,\curveA}=\seq{\leave{s,\curveA_1\curveA_2},\curveA_2,\ldots,\curveA_{m-1},\enter{t,\curveA_{m-1}\curveA_m}}$
\end{definition}

We remark that in the above definition $\clip{s,t,\curveA}$ is well defined as we assumed $m>2$ and thus $\leave{s,\curveA_{1}\curveA_2}$ must come before $\enter{t,\curveA_{m-1}\curveA_m}$ on $\curveA$. Moreover, we will later make use of the following observation. 

\begin{observation}\obslab{size}
Given a curve $\curveA=\seq{\curveA_1,\ldots,\curveA_m}$, a value $\thresh$, and points $s$ and $t$ such that $||s-\curveA_1||\leq \thresh$ and $||t-\curveA_m||\leq \thresh$, then $|\mv{s,t,\curveA}|\leq m$, since $\distFr{\curveA}{s\circ\curveA\circ t}\leq \thresh$.
\end{observation}

\cite{ghms-apsmlp-93} considered the problem of computing minimum link chains which stab an ordered sequence of disks. Under the right ordering and containment conditions, they argued an equivalence with \Frechet distance, thus yielding the following result. 
(Technically, Theorem 14 in \cite{ghms-apsmlp-93} does not specify starting and ending points, however, we show in \apndref{mvcreduction} there is a reduction from the problem of computing $\mv{s,t,\curveA}$  
to the same problem but where start and end vertices are not specified. Using \cite{ghms-apsmlp-93} for the latter dominates the reduction run time.)

\begin{theorem}[\cite{ghms-apsmlp-93}]\thmlab{stabtime}
Given $\curveA=\seq{\curveA_1,\ldots,\curveA_m}$, a value $\thresh$, 
and points $s$ and $t$ such that $||s-\curveA_1||\leq \thresh$ and $||t-\curveA_m||\leq \thresh$,
then $\mv{s,t,\curveA}$ can be computed in $O(m^2\log^2 m)$ time.
\end{theorem}


We have the following standard fact regarding the \Frechet distance.

\begin{fact}\factlab{ends}
 The \Frechet distance between two line segments is realized either at the starting or ending vertices. That is, for any $p,p',q,q'\in \Re^2$, $\distFr{\seq{p,p'}}{\seq{q,q'}}=\max\{||p-p'||,||q-q'||\}$.  
\end{fact}

The above fact leads to the following observation, which we will make use of later.

\begin{observation}\obslab{empty}
 Given $p,p',q,q'\in \Re^2$, by  \factref{ends}, $\distFr{\seq{p,p'}}{\seq{q,q'}}=\max\{||p-p'||,||q-q'||\}$.
 This implies, for any $k$ and any $q_1,\ldots,q_k\in \Re^2$, 
 $\distFr{\seq{p,p'}}{\seq{q,q_1,\ldots,q_k,q'}} \geq \max\{||p-p'||,||q-q'||\}= \distFr{\seq{p,p'}}{\seq{q,q'}}$. Thus if $||q-p||\leq \thresh$ and $||q'-p'||\leq \thresh$, then $\mv{q,q',\seq{p,p'}} = \seq{}$, i.e.\ the empty curve.
 
 Given curves $\curveA=\seq{\curveA_1,\ldots,\curveA_m}$ and $\curveB=\seq{\curveB_1,\ldots,\curveB_n}$, let $\curveB'$ be the modification of $\curveB$ realizing $\iedistFr{\curveA}{\curveB}$. Then the above implies that if $\curveB'$ was in part constructed by inserting vertices between $\curveB_i$ and $\curveB_{i+1}$, then the portion of $\curveA$ that the subcurve of $\curveB'$ between $\curveB_i$ and $\curveB_{i+1}$ will map to contains at least one vertex of $\curveA$. Together with \obsref{size}, this in turn implies $\iedistFr{\curveA}{\curveB}=O(m)$.
\end{observation}

\begin{lemma}\lemlab{shortcut}
Let $\curveA=\seq{\curveA_1,\ldots,\curveA_m}$ where $m>2$, let $s$ and $t$ be points such that $||s-\curveA_1||\leq \thresh$ and $||t-\curveA_m||\leq \thresh$, and let $\curveA'=\clip{s,t,\curveA}$. 
Finally, let $\mv{s,t,\curveA}=\curveB$ and let 
$\mv{s,t,\curveA'}=\curveB'$.
Then $\distFr{\curveA}{s\circ \curveB'\circ t}\leq \thresh$ and $|\curveB'|\leq |\curveB|$.
%
\end{lemma}
\begin{proof}
By definition of $\mv{s,t,\curveA'}=\curveB'$, 
we know that $\distFr{\curveA'}{s\circ \curveB'\circ t}\leq \thresh$, i.e.\ there is a bijective mapping between traversals of the two curves such that paired points are within distance $\thresh$. Refer to such traversals as $\thresh$-realizing traversals.  
Then this immediately implies that $\distFr{\curveA}{s\circ \curveB'\circ t}\leq \thresh$, since for our $\thresh$-realizing traversals one can stand still at $s$ (resp.\ $t$) while on $\curveA$ we linearly traverse from $\curveA_1$ to $\leave{s,\curveA_1\curveA_2}$ (resp.\ from $\curveA_m$ back to $\enter{t,\curveA_{m-1}\curveA_m}$), and then after (resp.\ before) that follow the $\thresh$-realizing traversals of $\curveA'$ and $s\circ \curveB'\circ t$. Note that all paired points are still within distance $\thresh$ as by definition the line segment $\curveA_1\leave{s,\curveA_1\curveA_2}\subset B(s,\thresh)$ (resp.\ $\curveA_m\enter{t,\curveA_{m-1}\curveA_m}\subset B(t,\thresh)$).  


\begin{figure}[!h]
    \centering
    \begin{subfigure}[b]{0.3\textwidth}
        \centering
        \includegraphics[width=1.0\textwidth]{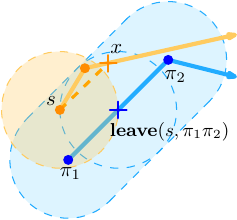}
        \subcaption{Shortcutting $\curveB$ to connect directly from $s$ to $x$.}
        \figlab{leavex}
    \end{subfigure}
    \hspace{1em}
    \begin{subfigure}[b]{0.3\textwidth}
        \centering
        \includegraphics[width=1.0\textwidth]{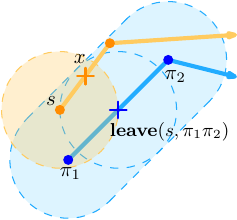}
        \subcaption{$x$ lies on the first segment.}
        \figlab{stayx}
    \end{subfigure}
    \hspace{1em}
    \begin{subfigure}[b]{0.3\textwidth}
        \centering
        \includegraphics[width=1.0\textwidth]{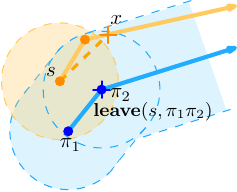}
        \subcaption{$\leave{s,\curveA_1\curveA_2}=\curveA_2$ and lies inside the $\delta$ ball around $s$}
        \figlab{insides}
    \end{subfigure}
    \caption{Different cases for $\leave{s,\curveA_1\curveA_2}$ and the point $x$ it gets paired with.}
    \figlab{fig:MinLink}
\end{figure}

We now prove that $|\curveB'|\leq |\curveB|$. To keep the cases to a minimum, let us focus on the $s$ side of $s\circ\curveB\circ t$. Specifically, let us temporarily redefine $\curveA'=\clip{s,t,\curveA}=\seq{\leave{s,\curveA_1\curveA_2},\curveA_2,\ldots,\curveA_{m}}$, and modify $\curveB'$ to be defined with respect to this new definition. Ultimately, the argument we make below will apply to the original $\curveA'=\clip{s,t,\curveA}$ and $\curveB'$ definitions by applying the argument at both ends of the curve.

Fix some $\thresh$-realizing traversal of $\curveA$ and $\curveB$, and let $x$ denote the point on $s\circ\curveB\circ t$ which $\leave{s,\curveA_1\curveA_2}$ from $\curveA$ gets paired with. (Recall that $\leave{s,\curveA_1\curveA_2}$ is the point in $B(s,\delta)\cap \curveA_1\curveA_2$ closest to $\curveA_2$, which may be $\curveA_2$ itself.) 
Let $\gamma$ denote the subcurve of $s\circ \curveB\circ t$ from $x$ to $t$.
Observe that $\leave{s,\curveA_1\curveA_2}$ is within distance $\thresh$ of both $s$ and $x$, and thus $sx\subset B(\leave{s,\curveA_1\curveA_2}, \thresh)$.
Therefore,
$\distFr{\curveA'}{s\circ \gamma}\leq \thresh$ 
since for the $\thresh$-realizing traversal of $\curveA'$ and $s\circ \gamma$ one can stand still at $\leave{s,\curveA_1\curveA_2}$ while on $s\circ \gamma$ traversing from $s$ to $x$, and then afterwards follow the portions of the $\thresh$-realizing traversal of $\curveA$ and $\curveB$ corresponding to their remaining subcurves.

So consider two cases. If $x$ lies on the first edge of $s\circ\curveB\circ t$, 
then the above fact that $\distFr{\curveA'}{s\circ \gamma}\leq \thresh$ implies that $\distFr{\curveA'}{s\circ \curveB\circ t}\leq \thresh$.\footnote{Roughly speaking, in this case $s\circ \gamma=s\circ\curveB \circ t$, though technically $s\circ \gamma$ has a vertex at $x$, where $x$ may not be a vertex on $s\circ\curveB\circ t$.}
Therefore, $|\curveB'|\leq |\curveB|$ since $\curveB'$ was the minimum vertex curve such that $\distFr{\curveA'}{s\circ \curveB'\circ t}\leq \thresh$.
So now suppose that $x$ does not lie on the first edge of $s\circ\curveB\circ t$ (see \figref{leavex}).
In this case observe that that $|s\circ \gamma|\leq |s\circ \sigma \circ t|$, since there must be at least one vertex between $s$ and $x$ on $\sigma$. 
However, $|s\circ \curveB'\circ t|\leq |s\circ \gamma|$, since again $\curveB'$ was the minimum vertex curve such that $\distFr{\curveA'}{s\circ \curveB'\circ t}\leq \thresh$, and thus again we conclude $|\curveB'|\leq |\curveB|$.
\end{proof}

\begin{definition}
    Given curves $\curveA=\seq{\curveA_1,\ldots,\curveA_m}$ and $\curveB=\seq{\curveB_1,\ldots,\curveB_n}$, we define the set of canonical inserted subcurves as
    \begin{equation*}
    \canon(\curveA,\curveB)=\left\{
      \mv{\curveB_i,\curveB_{i+1},\clip{\curveB_i,\curveB_{i+1}, \curveA[\alpha,\beta]}} ~~~\middle|~~~
       \begin{aligned}
         & i< n,~ \alpha<\beta-1\leq m-1, \\ & i,\alpha,\beta\in \mathbb{Z}^+, ~||\curveB_i-\curveA_{\alpha}\curveA_{\alpha+1}||\leq \thresh, \\ & ||\curveB_{i+1}-\curveA_{\beta-1}\curveA_\beta||\leq \thresh
       \end{aligned}
    \right\}
\end{equation*}
\end{definition}

Note that the curve $\curveB'$ realizing $\iedistFr{\curveA}{\curveB}$, contains $\curveB$ as a subsequence, where pairs of vertices that were consecutive in $\curveB$ either remain consecutive in $\curveB'$ or have a subcurve inserted between them. We thus have the following.

\begin{corollary}\corlab{canon}
    Given curves $\curveA=\seq{\curveA_1,\ldots,\curveA_m}$ and $\curveB=\seq{\curveB_1,\ldots,\curveB_n}$, where $\iedistFr{\curveA}{\curveB}\neq \infty$, then there exists a curve $\curveB'$ realizing $\iedistFr{\curveA}{\curveB}$, where the subcurve between any consecutive pair from $\curveB$ is in $\canon(\curveA,\curveB)$.
\end{corollary}
\begin{proof}
Let $\hat{\curveB}$ be any curve realizing $\iedistFr{\curveA}{\curveB}$.
Consider any consecutive pair $\curveB_i,\curveB_{i+1}$ from $\curveB$ where $\hat{\curveB}$ has vertices inserted between. That is, $\hat{\curveB}$ contains the subcurve $\curveB_i\circ\gamma \circ \curveB_{i+1}$ where $\gamma=\seq{\hat{\curveB}_{k_1},\ldots,\hat{\curveB}_{k_z}}$ and $z\geq 1$. Fix any $\delta$-realizing traversal of $\distFr{\curveA}{\hat{\curveB}}$, and under this traversal let $\curveA[x,y]$ be the subcurve of $\curveA$ that is matched to the subcurve $\curveB_i\circ\gamma \circ \curveB_{i+1}$ of $\hat{\curveB}$.
%
%
Note that $x$ and $y$ are not necessarily integers, i.e. $\curveA(x)$ and $\curveB(y)$ may lie on the interior of an edge.
Moreover, by \obsref{empty}, since $\hat{\curveB}$ actually paid to insert vertices to match $\curveA[x,y]$ it must be that $\curveA[x,y]$ is not a straight segment, i.e.\ it contains in its interior a vertex of $\curveA$. 
Due to this fact, and since $||\curveB_i-\curveA(x)||\leq \thresh$ and $||\curveB_{i+1}-\curveA(y)||\leq \thresh$, we know that both $\clip{\curveB_i,\curveB_{i+1}, \curveA[x,y]}$ and $\clip{\curveB_i,\curveB_{i+1}, \curveA[\lfloor x \rfloor, \lceil y \rceil]}$ are well defined, and in fact $\clip{\curveB_i,\curveB_{i+1}, \curveA[x,y]} = \clip{\curveB_i,\curveB_{i+1}, \curveA[\lfloor x \rfloor, \lceil y \rceil]}$.

Let $\gamma' =\mv{\curveB_i,\curveB_{i+1},\clip{\curveB_i,\curveB_{i+1}, \curveA[x,y]}}$, and note that the existence of $\gamma$ implies that by \lemref{shortcut} we have $\distFr{\curveA[x,y]}{\curveB_i\circ\gamma'\circ \curveB_{i+1}}\leq \thresh$, where $|\gamma'|\leq |\gamma|$, i.e.\ inserting $\gamma'$ between $\curveB_i$ and $\curveB_{i+1}$ instead of $\gamma$ is still an optimal solution. 
This completes the proof as $\gamma'=\mv{\curveB_i,\curveB_{i+1},\clip{\curveB_i,\curveB_{i+1}, \curveA[x,y]}}=\mv{\curveB_i,\curveB_{i+1},\clip{\curveB_i,\curveB_{i+1}, \curveA[\lfloor x \rfloor, \lceil y \rceil]}}$ and $\mv{\curveB_i,\curveB_{i+1},\clip{\curveB_i,\curveB_{i+1}, \curveA[\lfloor x \rfloor, \lceil y \rceil]}}\in \canon(\curveA,\curveB)$.
\end{proof}

\begin{remark}\remlab{extend}
The above discussion easily extends to the case where insertions are allowed before $\curveB_1$ and after $\curveB_n$.  First, we extend all the above definitions.

Let $\mv{s,\cdot,\curveA}$ (resp.\ $\mv{\cdot,t,\curveA}$) denote the curve $\curveB=\seq{\curveB_1,\ldots,\curveB_n}$ with the minimum number of vertices such that 
$\distFr{\curveA}{s\circ\curveB}\leq \delta$ (resp.\ 
$\distFr{\curveA}{\curveB\circ t}\leq \delta$). 
Given a curve $\curveA=\seq{\curveA_1,\ldots,\curveA_m}$ where $m>2$, and a point $s$ (resp.\ a point $t$) such that $||s-\curveA_1\curveA_2||\leq \thresh$ (resp.\  $||t-\curveA_{m-1}\curveA_m||\leq \thresh$), define $\clip{s,\cdot,\curveA}=\seq{\leave{s,\curveA_1\curveA_2},\curveA_2,\ldots,\curveA_{m-1},\curveA_m}$ (resp.\ $\clip{\cdot,t,\curveA}=\seq{\curveA_1,\curveA_2,\ldots,\curveA_{m-1},\enter{t,\curveA_{m-1}\curveA_m}}$). Finally, define the set of canonical subcurves which end at $\curveB_1$ or start at $\curveB_n$ as follows.
    \begin{equation*}
    \canon^e_1(\curveA,\curveB)=\left\{
      \mv{\cdot,\curveB_{1},\clip{\cdot,\curveB_{1}, \curveA[1,\beta]}} ~~~\middle|~~~
         \beta\leq m, \beta\in \mathbb{Z}^+, ||\curveB_{1}-\curveA_{\beta-1}\curveA_\beta||\leq \thresh
    \right\}
\end{equation*}
    \begin{equation*}
    \canon^s_n(\curveA,\curveB)=\left\{
      \mv{\curveB_{n},\cdot,\clip{\curveB_{n},\cdot, \curveA[\alpha,n]}} ~~~\middle|~~~
         1 \leq \alpha, \alpha\in \mathbb{Z}^+, ||\curveB_{n}-\curveA_{\alpha}\curveA_{\alpha+1}||\leq \thresh
    \right\}
\end{equation*}
Note that \thmref{stabtime} extends to also compute $\mv{s,\cdot,\curveA}$ and $\mv{\cdot,t,\curveA}$ in $O(m^2\log n)$ time.
(Indeed, as mentioned above, the result from \cite{ghms-apsmlp-93} did not specify the start and end points, and a reduction in \apndref{mvcreduction} is given to extend to this case.) 
Moreover, \corref{canon} extends to the sets $\canon^e_1(\curveA,\curveB)$ and $\canon^s_n(\curveA,\curveB)$ when considering the portions of the curve before $\curveB_1$ and after $\curveB_n$, respectively. Indeed the proof is arguably easier in these cases, as the portion of $\curveA$ that is being mapped to must contain $\curveA_1$ or $\curveA_m$, respectively.
\end{remark}

\subsubsection{The Algorithm}

Given curves $\curveA=\seq{\curveA_1,\ldots,\curveA_m}$ and $\curveB=\seq{\curveB_1,\ldots,\curveB_n}$, a  threshold $\thresh$, and a parameter $k$, we now describe how to determine if $\iedistFr{\curveA}{\curveB}\leq k$. 
Ultimately, as $k=O(m)$, this also yield an algorithm to compute $\iedistFr{\curveA}{\curveB}$, analgously to how our bound on $k$ yielded an algorithm to compute $\dedistFr{\curveA}{\curveB}$.

We will follow a similar approach to that in \secref{cdo}. As $\curveA$ will remain unchanged again we use $\curveA$ itself as a \DAG complex. For $\curveB$, we again create $k$ additional copies of $\curveB$ with $\curveB^{\ell}=\seq{\curveB_1^{\ell},\ldots,\curveB_n^{\ell}}$ denoting the $\ell$th copy. Now for the edges, for deletion only we used the complete weighted \DAG complex which for any $0\leq \ell\leq k$ and any $i<j$ such $\ell+(j-(i+1)) \leq k$, included the directed edge $\curveB_i^{\ell}\curveB_j^{\ell+(j-(i+1))}$ to encode deletion of the $j-(i+1)$ vertices between $\curveB_i$ and $\curveB_j$. 
Thus now to instead encode inserting vertices between consecutive pairs on $\curveB$, for each curve $\gamma\in \canon(\curveA,\curveB)$, where $\gamma$ connected from $\curveB_i$ to $\curveB_{i+1}$, we add the curve $\seq{\curveB_i^\ell}\circ \gamma\circ \seq{\curveB_{i+1}^{\ell+|\gamma|}}$ to the \DAG complex (where $|\gamma|$ denotes its number of vertices), for all $\ell$ such that $0\leq \ell<\ell+|\gamma|\leq k$.
We also need to account for the possibility that vertices get inserted before $\curveB_1$ or after $\curveB_n$, which by \remref{extend}, is handled by including the curves from $\canon^e_1(\curveA,\curveB)$ and $\canon^s_n(\curveA,\curveB)$, respectively. Specifically, for a curve $\gamma\in \canon^e_1(\curveA,\curveB)$, we identify the last vertex of $\gamma$ (which by definition is located at $\curveB_1$) with $\curveB_1^{|\gamma|-1}$, so long as $|\gamma|-1\leq k$. (Note it is $|\gamma|-1$ and not $|\gamma|$, as we are not inserting $\curveB_1$.) Now for each $\gamma\in \canon^s_n(\curveA,\curveB)$ we need to add multiple copies since we don't know the cost to reach $\curveB_n$. Specifically, we identify the first vertex of a copy of $\gamma$ with $\curveB_n^{\ell + |\gamma|-1}$, for all $\ell$ such that $0\leq \ell+|\gamma|-1\leq k$.

Call the above described complex, the insertion weighted complex of $\curveB$ with respect to $\curveA$, denoted $\complex_{\curveB,\curveA}$, and observe that by \corref{canon} (and \remref{extend}) we know that if $\iedistFr{\curveA}{\curveB}\leq k$ then there is a curve $\curveB'$ realizing $\iedistFr{\curveA}{\curveB}$ which is compliant with $\complex_{\curveB,\curveA}$.
So let $V(\canon^e_1(\curveA,\curveB))$ denote the set of vertices consisting of the first vertex for each curve $\gamma\in \canon^e_1(\curveA,\curveB)$ that we included in $\complex_{\curveB,\curveA}$.
Then $S_\curveB=\{\curveB_1^0\}\cup V(\canon^s_n(\curveA,\curveB))$.
Analogously let $V(\canon^s_n(\curveA,\curveB))$ denote the set of vertices consisting of the last vertex for every copy each curve $\gamma\in \canon^s_n(\curveA,\curveB)$ that we included in $\complex_{\curveB,\curveA}$.
Then $T_\curveB=\{\curveB_n^\ell \mid 0\leq \ell\leq k\}\cup V(\canon^s_n(\curveA,\curveB))$.

Again similar to \secref{cdo}, we call \thmref{dag}, which computes the set of all pairs in $\curveA_m\times T_\curveB$ such that there are compliant paths from allowable starting vertices whose \Frechet distance is $\leq \thresh$. If no such pair exists then $\iedistFr{\curveA}{\curveB}> k$. 
So let $x\in T_\curveB$ be such that $(\curveA_m, x)$ is reachable. If $x=\curveB_n^{\alpha}$ for some $\alpha$, then it corresponds to inserting $\alpha$ vertices. Otherwise, $x\in V(\canon^s_n(\curveA,\curveB))$ and is the last vertex on copy of a curve $\gamma\in \canon^s_n(\curveA,\curveB)$ which was attached at some $\curveB_n^\ell$, in which case it corresponds to inserting $\alpha=\ell+|\gamma|-1$ vertices.
In either case, such a vertex was only included in $\complex_{\curveB,\curveA}$ if $\alpha\leq k$, and so we can conclude that $\iedistFr{\curveA}{\curveB}\leq k$.

As for the running time, computing $\canon(\curveA,\curveB)$ takes $O(nm^4\log^2 m)$ time as $|\canon(\curveA,\curveB)|=O(nm^2)$ and each curve in $\canon(\curveA,\curveB)$ takes $O(m^2\log^2 m)$ time to compute using  \thmref{stabtime}.\footnote{
We conjecture that the time to compute $\canon(\curveA,\curveB)$ may be improved by reusing and updating the computations from the algorithm in \cite{ghms-apsmlp-93}, rather than making independent calls.}
(Note this also dominates the time to compute the smaller sets $\canon^e_1(\curveA,\curveB)$ and $\canon^s_n(\curveA,\curveB)$.)
The complex constructed for $\curveA$ has size $O(m)$, and the insertion weighted complex constructed for $\curveB$ has size $O(k^2nm^2)$.
Thus the total time is $O(nm^3(k^2 +m\log^2 m))$. 
Recall that by \obsref{empty}, 
$\iedistFr{\curveA}{\curveB}=O(m)$.
Therefore, $O(nm^3(k^2 +m\log^2 m))=O(nm^5)$.
}
\RegVer{\continuousinsert}
\SOCGVer{
For insertions only, instead of adding  directed edges between pairs of vertices of copies of \(\curveB\), we compute and insert copies of the following \(O(nm^2)\) \emph{canonical subcurves}.
\begin{equation*}
    \canon(\curveA,\curveB)=\left\{
      \mv{\curveB_i,\curveB_{i+1},\clip{\curveB_i,\curveB_{i+1}, \curveA[\alpha,\beta]}} ~~~\middle|~~~
       \begin{aligned}
         & i< n,~ \alpha<\beta-1\leq m-1, \\ & i,\alpha,\beta\in \mathbb{Z}^+, ~||\curveB_i-\curveA_{\alpha}\curveA_{\alpha+1}||\leq \thresh, \\ & ||\curveB_{i+1}-\curveA_{\beta-1}\curveA_\beta||\leq \thresh
       \end{aligned}
    \right\}
\end{equation*}
Computing $\canon(\curveA,\curveB)$ takes $O(nm^4\log^2 m)$ time.
The complex for $\curveA$ has size $O(m)$, and the insertion weighted complex constructed for $\curveB$ has size $O(k^2nm^2)$. Thus the total time to construct the complexes and find nearby curves within them is $O(nm^3(k^2+m\log^2 m))$. 
We can argue that
$\iedistFr{\curveA}{\curveB}=O(m)$.
Therefore, $O(nm^3(k^2+m\log^2 m))=O(nm^5)$.
Deletions and insertions together are handled similarly by extending $\canon(\curveA,\curveB)$ to be defined over all pairs on $\curveB$.
See the full version for details.
}

\begin{theorem}\thmlab{continuousinsert}
Given curves $\curveA=\seq{\curveA_1,\ldots,\curveA_m}$ and $\curveB=\seq{\curveB_1,\ldots,\curveB_n}$ in $\Re^2$, a threshold $\thresh$, and an integer $k>0$, in $O(nm^3(k^2+m\log^2 m))$ time one can determine if $\iedistFr{\curveA}{\curveB}\leq k$. 
Moreover, 
one can compute $\iedistFr{\curveA}{\curveB}$ in $O(nm^5)$ time.
\end{theorem}

\newcommand{\continuousboth}{
\subsection{Insertion and Deletion}\label{cbapnd}

We can now easily allow for both insertions and deletions by combining the above approaches. For deletion only we allowed connecting between any pair of vertices from $\curveB$ with a segment, whereas for the insertion only case we added paths from $\canon(\curveA,\curveB)$ between adjacent vertices from $\curveB$. Thus we now extend the definition of $\canon(\curveA,\curveB)$ to all pairs from $\curveB$.

\begin{definition}
Given curves $\curveA=\seq{\curveA_1,\ldots,\curveA_m}$ and $\curveB=\seq{\curveB_1,\ldots,\curveB_n}$, we define the extended set of canonical inserted subcurves as
    \begin{equation*}
    \ecanon(\curveA,\curveB)=\left\{
      \mv{\curveB_i,\curveB_{j},\clip{\curveB_i,\curveB_{j}, \curveA[\alpha,\beta]}} ~~\middle|~~
       \begin{aligned}
         & i<j\leq n, ~ \alpha<\beta-1 \leq m-1, \\ & i,j,\alpha,\beta\in \mathbb{Z}^+, ~||\curveB_i-\curveA_{\alpha}\curveA_{\alpha+1}||\leq \thresh, \\ & ||\curveB_{j}-\curveA_{\beta-1}\curveA_\beta||\leq \thresh
       \end{aligned}
    \right\}
\end{equation*}
\end{definition}
We analogously extend the definitions of $\canon^e_1(\curveA,\curveB)$ and and $\canon^s_n(\curveA,\curveB)$ from \remref{extend} to the sets $\ecanon^e(\curveA,\curveB) = \bigcup_{i} \canon^e_i(\curveA,\curveB)$ and $\ecanon^s(\curveA,\curveB) = \bigcup_{i} \canon^s_i(\curveA,\curveB)$, respectively.

Again we use $\curveA$ itself as a \DAG complex. For $\curveB$, we again create $k$ additional copies of $\curveB$ with $\curveB^{\ell}=\seq{\curveB_1^{\ell},\ldots,\curveB_n^{\ell}}$ denoting the $\ell$th copy. 
For each curve $\gamma\in \ecanon(\curveA,\curveB)$, where $\gamma$ connected from $\curveB_i$ to $\curveB_{j}$, we add the curve $\seq{\curveB_i^\ell}\circ \gamma\circ \seq{\curveB_{j}^{\ell+|\gamma|+(j-(i+1))}}$ to the \DAG complex, for all $\ell$ such that $0\leq \ell<\ell+|\gamma|+(j-(i+1))\leq k$. 
For each curve $\gamma\in \canon^e_i(\curveA,\curveB)$, we identify the last vertex of $\gamma$ with $\curveB_i^{(i-1)+(|\gamma|-1)}=\curveB_i^{i+|\gamma|-2}$, so long as $i+|\gamma|-2\leq k$. 
Now for each $\gamma\in \canon^s_i(\curveA,\curveB)$ we need to add multiple copies since we don't know the cost to reach $\curveB_i$. Specifically, we identify the first vertex of a copy of $\gamma$ with $\curveB_i^{\ell + (n-i) +|\gamma|-1}$, for all $\ell$ such that $0\leq \ell+(n-i)+ |\gamma|-1\leq k$.

Let $\complex_\curveB$ denote the resulting \DAG complex.
Note that one can view a set of edits to $\curveB$, as first making deletions and then making insertions.
Thus by the arguments in \secref{cdo} and \secref{cio}, if $\edistFr{\curveA}{\curveB}\leq k$ then there is a curve $\curveB'$ realizing $\edistFr{\curveA}{\curveB}$ which is compliant with $\complex_{\curveB}$.
Now for $\curveB$ the optimal solution may delete some prefix of vertices $\curveB_1,\ldots,\curveB_{i}$, thus 
$S_\curveB=\bigcup_i \{\{\curveB_{i+1}^i\}\cup V(\canon^e_i(\curveA,\curveB))\}$
Similarly, the optimal solution may delete some suffix of vertices from $\curveB$, and so $T_\curveB=\bigcup_i \{\{\curveB_i^\ell \mid 0\leq \ell\leq k\}\cup V(\canon^s_i(\curveA,\curveB))\}$. 

%
Calling \thmref{dag} computes the set of all pairs in $\curveA_m\times T_\curveB$ such that there are compliant paths from allowable starting vertices whose \Frechet distance is $\leq \thresh$.
We also check if $n + |\mv{\curveA}| \leq k$ to account for the extreme case that it suffices to delete all vertices of $\curveB$ and replace them with $\mv{\curveA}$.
If no such pair in $\curveA_m \times T_\curveB$ exists, and our additional check for complete replacement fails, then $\edistFr{\curveA}{\curveB}> k$. 
So let $x\in T_\curveB$ be such that $(\curveA_m, x)$ is reachable. If $x=\curveB_i^{\ell}$ for some $i$ and $\ell$, then it corresponds to $\ell$ edits to reach $\curveB_i$ followed by deleting $n-i$ vertices after $\curveB_i$, so $\alpha=\ell+(n-i)$ edits overall.
Otherwise, $x\in V(\canon^s_i(\curveA,\curveB))$, for some $i$, and is the last vertex on a copy of a curve $\gamma\in \canon^s_i(\curveA,\curveB)$ which was attached at some $\curveB_i^\ell$, in which case it corresponds to $\ell$ edits followed by $(n-i)$ deletions and $|\gamma|-1$ insertions, so $\alpha=\ell+(n-i)+|\gamma|-1$ edits overall. 
Thus for any $x\in T_\curveB$, if $\alpha\leq k$ then $\edistFr{\curveA}{\curveB}\leq k$, and if $\alpha> k$ for all $x\in T_\curveB$ (and our complete replacement check fails) then $\edistFr{\curveA}{\curveB}> k$.


As for the running time, computing $\ecanon(\curveA,\curveB)$ takes $O(n^2m^4\log^2 m)$ time as $|\ecanon(\curveA,\curveB)|=O(n^2m^2)$ and each curve in $\ecanon(\curveA,\curveB)$ takes $O(m^2\log^2 m)$ time to compute using  \thmref{stabtime}.
However, observe that if we limit ourselves to $k$ edits, then we only need to compute the subset of $\ecanon(\curveA,\curveB)$ between pairs $\curveB_i, \curveB_j$ such that $j-(i+1)\leq k$, yielding a time of $O(knm^4\log^2 m)$.
(As before, this time is also sufficient to compute all $\canon^e_i(\curveA,\curveB)$.)
%
The complex constructed for $\curveA$ has size $O(m)$, and the complex constructed for $\curveB$ has size $O(k^3nm^2)$. Thus the total time to construct the complexes and find nearby curves within them is $O(knm^3(k^2+m\log^2 m))$. 
As discussed above, $k=O(m+n)$.
}
\RegVer{\continuousboth}

\begin{theorem}\thmlab{continuousboth}
Given curves $\curveA=\seq{\curveA_1,\ldots,\curveA_m}$ and $\curveB=\seq{\curveB_1,\ldots,\curveB_n}$ in $\Re^2$, a threshold $\thresh$, and an integer $k>0$, in $O(knm^3(k^2+m\log^2 m))$ time one can determine if $\edistFr{\curveA}{\curveB}\leq k$.
Moreover, one can compute $\edistFr{\curveA}{\curveB}$ in $O((m+n)^3nm^3)$ time.
\end{theorem}


\newcommand{\subcont}{

\subsection{Substitutions}
\br{If we include this subsection, we need to modify all the text before it in the paper to say we are also considering substitution, and move this to appendix.}

Here we show how the approach above easily extends to the case where in addition to deletions and insertions, we also allow the edit operation of substitutions. First, observe that a substitution itself can be viewed as a deletion followed by an insertion, although the cost is now 1 edit not 2. As our approach in Section~\ref{cbapnd} was to model the different possible combinations of insertions and deletions of cost at most $k$, the same approach will work if we include substitutions but with different costs (or more precisely connecting to different layers in the \DAG complex).

More precisely, let $\gamma\in \ecanon$ be a curve such that we added $\seq{\curveB_i^\ell}\circ \gamma\circ \seq{\curveB_{j}^{\ell+|\gamma|+(j-(i+1))}}$ to the \DAG complex, for some $\ell$ such that $0\leq \ell<\ell+|\gamma|+(j-(i+1))\leq k$.  This represented deleting the $j-(i+1)$ vertices $\curveB_{i+1},\ldots,\curveB_{j-1}$ and inserting the $|\gamma|$ vertices of $\gamma$, for a total cost of $j-(i+1)+|\gamma|$. 
So let $\alpha=\min\{|\gamma|, j-(i+1)\}$. Then instead of first deleting $\curveB_{i+1},\ldots,\curveB_{i+\alpha}$ and then inserting $\gamma_1,\ldots, \gamma_\alpha$, instead we can substitute $\curveB_{i+x}$ with $\gamma_x$ for all $1\leq x\leq \alpha$, for a total saving of $\alpha$. This implies that instead of having $\gamma$ connect from  $\curveB_i^\ell$ to $\curveB_{j}^{\ell+|\gamma|+(j-(i+1))}$ in the \DAG complex, instead it should connect from $\curveB_i^\ell$ to $\curveB_{j}^{\ell+|\gamma|+(j-(i+1))-\alpha}$, where $\alpha=\min\{|\gamma|, j-(i+1)\}$. (Note this means will also insert some additional copies of $\gamma$ that were previously to expensive, as now we consider all $\ell$ such that $0\leq \ell<\ell+|\gamma|+(j-(i+1))-\alpha \leq k$.)

Thus we have the following analogue of \thmref{continuousboth}.

\br{Does $\edistFr{}{}$ now refer to insert+delete, or also substitution. This may mean updating notation in prior sections. Also for the theorem, when plugging in for $k$ we might get a slightly better bound as I believe k is at most $\max{n,m}$ rather than $n+m$.}

\begin{theorem}\thmlab{continuousall}
Given curves $\curveA=\seq{\curveA_1,\ldots,\curveA_m}$ and $\curveB=\seq{\curveB_1,\ldots,\curveB_n}$ in $\Re^2$, a threshold $\thresh$, and an integer $k>0$, in $O(knm^3(k+m\log^2 m))$ time one can determine if $\edistFr{\curveA}{\curveB}\leq k$.
Moreover, one can compute $\edistFr{\curveA}{\curveB}$ in $O((m+n)nm^3(n+m\log^2 m))$ time.
\end{theorem}
}


\section{Discrete \Frechet Distance}

We now discuss the discrete analogs \(\dedistDFr{\curveA}{\curveB}\), \(\iedistDFr{\curveA}{\curveB}\), and \(\edistDFr{\curveA}{\curveB}\) of the problems in the previous section.
The extra structure afforded by considering discrete point sequences allows us to more directly apply standard dynamic programming techniques and achieve faster running times for all three problems and in any constant dimension.

\subsection{Deletion Only}
The deletion only variant \(\dedistDFr{\curveA}{\curveB}\) serves as an easy warm up.
Let \(\dedistdp{i}{j} := \dedistDFr{\curveA[1,i]}{\curveB[1, j]}\) (with \(i = 0\) and \(j = 0\) denoting empty prefixes, and  \(\dedistdp{0}{0}=0\)).
Suppose there is a set of deletions changing \(\curveB[1,j]\) into a curve \(\curveB'\) such that \(\distDFr{\curveA[1,i]}{\curveB'} \leq \thresh\).

If \(i \geq 1\), then we must have \(j \geq 1\) as well.
Suppose further that \(||\curveB_j - \curveA_i|| \leq \thresh\).
Now, any monotone correspondence between \(\curveA[1,i]\) and \(\curveB'\) already includes or can be extended to include the pair \((\curveA_i,\curveB_j)\) without increasing the maximum distance of a pair beyond \(\thresh\).
Therefore, we may assume \(\curveB'\) ends with \(\curveB_j\).
As in the normal dynamic programming solution for the discrete \Frechet distance, we may further assume the rest of the correspondence matches all of curves \(\curveA[1,i]\) and \(\curveB'\) except for the last point of one or both of them.

If \(i = 0\) and \(j\geq 1\) then clearly $\curveB_j$ must be deleted as there is no vertex of $\curveA$ to match it to. Similarly, if $i,j\geq 1$ and \(||\curveB_j - \curveA_i|| > \thresh\), then again $\curveB_j$ must be deleted as all monotone correspondences between \(\curveA[1,i]\) and \(\curveB'\) end with a pair containing the last point of both.

From the above discussion, we conclude
\[ \dedistdp{i}{j} =
\begin{cases}
    0 &\text{if \(i=0\) and \(j = 0\)}\\
    \infty &\text{if \(i \geq 1\) and \(j = 0\)}\\
    1 + \dedistdp{i}{j-1} &\text{if (\(i = 0\) and \(j \geq 1\))}\\
    &\quad\text{or (\(i,j \geq 1\) and \(||\curveB_j - \curveA_i|| > \thresh\))}\\

    \min \left\{
        \begin{aligned}
        \dedistdp{i}{j-1},\\
        \dedistdp{i-1}{j},\\
        \dedistdp{i-1}{j-1}
        \end{aligned}\right\} &\text{otherwise}
\end{cases}.\]
\(\dedistDFr{\curveA}{\curveB} = \dedistdp{m}{n}\) can be computed easily in \(O(mn)\) time using this recurrence.
\begin{theorem}
    Given curves $\curveA=\seq{\curveA_1,\ldots,\curveA_m}$ and $\curveB=\seq{\curveB_1,\ldots,\curveB_n}$ in \(\mathbb{R}^d\) and a threshold $\thresh$, one can compute $\dedistDFr{\curveA}{\curveB}$ in $O(mn)$ time.
\end{theorem}

\RegVer{\subsection{Insertions Only}}

\SOCGVer{\subsection{Insertions}\seclab{dinsert}}

We now consider the insertion only variant \(\iedistDFr{\curveA}{\curveB}\).
Let \(\iedistdp{i}{j} := \iedistDFr{\curveA[1,i]}{\curveB[1, j]}\).
As before, assume there is a set of insertions changing \(\curveB[1,j]\) to \(\curveB'\) where \(\distDFr{\curveA[1,i]}{\curveB'}\leq \thresh\).

Suppose \(\curveB'\) ends with \(\curveB_j\), implying \(||\curveB_j - \curveA_i|| \leq \thresh\).
(It is important to note for later that if \(||\curveB_j - \curveA_i|| \leq \thresh\) it does not imply \(\curveB'\) ends with \(\curveB_j\).)
We get the three standard cases for computing the discrete \Frechet distance as before.

Now suppose \(\curveB'\) does not end with \(\curveB_j\) and instead ends with a newly inserted point.
Let \(x\) denote this final point of \(\curveB'\).
There exists some \(k \in \{1, \dots, i\}\) such that the monotone correspondence with maximum distance at most \(\thresh\) between \(\curveB'\) and \(\curveA[1,i]\) ends with pairs between points of \(\seq{\curveA_k, \dots, \curveA_i}\) and \(x\).
These points of \(\seq{\curveA_k, \dots, \curveA_i}\) all live in \(\ball(x,\thresh)\), the ball of radius \(\thresh\) centered at \(x\).
%
Accordingly, let \(\suffixmeb(i)\) denote the smallest \(t \in \{1, \dots, i\}\) such that the radius of the minimum enclosing ball of \(\seq{\curveA_t, \dots, \curveA_i}\) is at most \(\thresh\).
We may assume \(x\) is the center of the ball defining \(\suffixmeb(i)\) and that \(\suffixmeb(i) \leq k \leq i\).
We have the following recurrence. 
\[ \iedistdp{i}{j} =
\begin{cases}
    0 &\text{if \(i = 0\) and \(j =0\)}\\
    \infty &\text{if \(i = 0\) and \(j \geq 1\)}\\
    1 + \min_{\suffixmeb(i) \leq k \leq i} \iedistdp{k-1}{j} &\text{if (\(i \geq 1\) and \(j = 0\))}\\
    &\quad\text{or (\(i,j \geq 1\) and \(||\curveB_j - \curveA_i|| > \thresh\))}\\
    \min\left\{
        \begin{aligned}
            \iedistdp{i}{j-1},\\
            \iedistdp{i-1}{j},\\
            \iedistdp{i-1}{j-1},\\
            1 + \min_{\suffixmeb(i) \leq k \leq i} \iedistdp{k-1}{j}
        \end{aligned}
    \right\} &\text{otherwise}
\end{cases}. \]
\SOCGVer{After \(O(m^2)\) preprocessing time and through careful use of simple data structures, we are able to solve all the relevant subproblems in \(O(mn)\) time.
A similar recurrence and dynamic programming strategy works for \(\edistDFr{\curveA}{\curveB}\).
See the full version for details.}

\newcommand{\discreteinsertonly}{
To efficiently implement the dynamic programming algorithm for insertions, we will first require some lemmas. 

\begin{lemma}
\label{lem:compute_suffixmeb}
We can compute \(\suffixmeb(i)\) for all \(i \in \{1, \dots, m\}\) in \(O(m^2)\) time assuming the dimension \(d\) is a constant.
\end{lemma}
\begin{proof}
We first observe that $\suffixmeb(1)\leq \ldots \leq \suffixmeb(m)$, because any ball enclosing \(\seq{\curveA_t, \dots, \curveA_i}\) also encloses \(\seq{\curveA_t, \dots, \curveA_{i-1}}\).

We compute the individual \(\suffixmeb(i)\) in increasing order of \(i\).
Suppose that we have computed $\suffixmeb(1), \ldots, \suffixmeb(i-1)$ and we are about to compute $\suffixmeb(i)$.
Set $t := \suffixmeb(i-1)$, and consider the minimum enclosing ball of $p_t, \ldots, p_i$.
If the radius is at most \(\thresh\) then \(\suffixmeb(i)\) being non-decreasing in \(i\) implies \(\suffixmeb(i) = t\).
If the radius is greater than \(\thresh\), then we conclude that $\suffixmeb(i) > t$.
Accordingly, we compute the radii of minimum enclosing balls for \(\seq{p_{t'}, \dots, p_i}\) for each \(t' \geq t\) until we find the smallest \(t'\) such that the radius is at most \(\thresh\).

Overall, our algorithm computes $O(m)$ minimum enclosing balls, because each time we compute a ball, we either increase $t'$ or $i$.
A single minimum enclosing ball over \(m\) points in \(\mathbb{R}^d\) can be computed in \(\Theta(m)\) time~\cite{m-lpltd-84,c-idalp-18}.
Therefore, we spend $O(m^2)$ time in total. 
\end{proof}

Naively, most of the \(O(mn)\) subproblems require \(\Omega(m)\) time to solve, even after precomputing all \(\suffixmeb(i)\).
However, we can take advantage of the following result that we believe is best attributed to forklore.

\begin{lemma}
    \label{lem:queue_with_min}
    Given a universe of elements with priorities, one can augment a standard first-in-first-out queue so that it can return its minimum priority element.
    Finding the minimum priority element and dequeing from the front of the queue take \(O(1)\) time in the worst case.
    Enqueuing a new element in the back of the queue takes \(O(1)\) amortized time.
    \end{lemma}
    \begin{proof}
    In addition to the normal queue data structure, we keep an additional doubly-linked list of elements that may at some point in the future become the element of minimum priority.
    The head of the list is the minimum priority element.
    Each element \(e\) in the list is succeeded by the minimum priority element inserted after \(e\).
    To find the minimum priority element of the whole queue at any time, we simply return the head of the list.
    To dequeue, we remove the element from the queue, and if this element was also the element at the head of the list then we additionally delete the element at the head of the list.
    Finally, to enqueue an element \(e\), we search the list in backwards order starting from its tail, deleting each element of priority greater than \(e\) until either the list becomes empty or we find an element \(e'\) of priority less than or equal to that of \(e\).
    In the former case, \(e\) becomes the head and sole member of the list.
    In the latter case, \(e'\) is succeeded by \(e\).
    
    The first two operations take \(O(1)\) time in the worst case.
    Each element can be removed from the list at most once, so enqueing takes \(O(1)\) amortized time.
    \end{proof}

We compute all \(\iedistdp{i}{j}\) in \(j\)-major order.
Fix any \(j\).
To compute the \(\iedistdp{i}{j}\), we create a new instance of Lemma~\ref{lem:queue_with_min}'s data structure.
Suppose we have just computed \(\iedistdp{i-1}{j}\).
We assume inductively that the queue contains as its elements all \(k \in \{\suffixmeb(i-1), \dots, i - 1\}\) with \(\suffixmeb(i-1)\) at the front where each \(k\) has priority \(\iedistdp{k-1}{j}\).
We dequeue all \(k \in \{\suffixmeb(i-1), \dots, \suffixmeb(i) - 1\}\) and enqueue \(i\) with priority \(\iedistdp{i-1}{j}\).
We can now evaluate all the cases for our fixed \(j\) in \(O(m)\) time total using the data structure.
}

\RegVer{
\discreteinsertonly
Considering all of the above, we conclude the following.}
\begin{theorem}\thmlab{discreteinsertion}
    Given curves $\curveA=\seq{\curveA_1,\ldots,\curveA_m}$ and $\curveB=\seq{\curveB_1,\ldots,\curveB_n}$ in \(\mathbb{R}^d\) for constant \(d\) and a threshold $\thresh$, one can compute $\iedistDFr{\curveA}{\curveB}$ in $O(m^2 + mn)$ time.
\end{theorem}

\newcommand{\discreteboth}{
\subsection{Insertion and Deletion}

Handling both insertions and deletions can be done by simply combining the two sets of cases described previously.
Let \(\edistdp{i}{j} := \edistDFr{\curveA[1,i]}{\curveB[1,j]}\).
The \(\infty\) base cases are avoided, because we can always delete the last point of \(\curveB[1,i]\) or insert a new point into the empty curve \(\curveB[1,0]\).
\[ \edistdp{i}{j} =
\begin{cases}
    0 &\text{if \(i = j = 0\)}\\
    1 + \edistdp{i}{j-1} &\text{if \(i = 0\) and \(j \geq 1\)}\\
    1 + \min_{\suffixmeb(i) \leq k \leq i} \edistdp{k-1}{j} &\text{if \(i \geq 1\) and \(j = 0\)}\\
    \min\left\{
        \begin{aligned}
            1 + \edistdp{i}{j-1},\\
            1 + \min_{\suffixmeb(i) \leq k \leq i} \edistdp{k-1}{j}
        \end{aligned}
    \right\} &\text{if \(i,j \geq 1\) and \(||\curveB_j - \curveA_i|| > \thresh\)}\\
    \min\left\{
        \begin{aligned}
            \edistdp{i}{j-1},\\
            \edistdp{i-1}{j},\\
            \edistdp{i-1}{j-1},\\
            1 + \min_{\suffixmeb(i) \leq k \leq i} \edistdp{k-1}{j}
        \end{aligned}
    \right\} &\text{otherwise}
\end{cases}. \]

We again use Lemma~\ref{lem:compute_suffixmeb} to evaluate each $\suffixmeb(i)$ quickly, and the data structure of Lemma~\ref{lem:queue_with_min} to evaluate the subproblems quickly.
}
\RegVer{\discreteboth}

\begin{theorem}\thmlab{discreteboth}
Given curves $\curveA=\seq{\curveA_1,\ldots,\curveA_m}$ and $\curveB=\seq{\curveB_1,\ldots,\curveB_n}$ in \(\mathbb{R}^d\) for constant \(d\) and a threshold $\thresh$, one can compute $\edistDFr{\curveA}{\curveB}$ in $O(m^2 + mn)$ time.
\end{theorem}

\section{Hardness}
\label{sec:hardness}
\newcommand{\SATI}{I}    
\newcommand{\SATv}{v} 
\newcommand{\SATc}{c} 
\newcommand{\BoolVar}{\mathcal{X}}  
\newcommand{\YES}{CHANGE}  
\newcommand{\FALSE}{\textsf{False}\xspace}
\newcommand{\TRUE}{\textsf{True}\xspace}

In this section we prove that a number of variants of the weak edit \Frechet distance are NP-hard. For these variants we will first focus on the discrete \Frechet distance case, showing NP-hardness even when the curves are restricted to points in $\Re^1$. Afterwards we show how the NP-hardness proofs easily extend to the continuous case for curves in $\Re^2$. All the NP-hardness proofs will be by a reduction from 3SAT, inspired by the reduction in \cite{blopuv-cfducod-23}.

\RegVer{
First, we prove NP-hardness of weak discrete edit \Frechet distance, where edits are restricted to deletions ($\dedistDFrW{\curveA}{\curveB}$). For this case we prove the problem is NP-hard with unlimited deletions on one curve, as well as limited deletions on one or both curves. We then show weak edit \Frechet distance restricted to limited insertions on one curve ($\iedistDFrW{\curveA}{\curveB}$) is NP-hard, which easily combines with the prior findings to show that limited insertions and deletions on one curve ($\edistDFrW{\curveA}{\curveB}$) is also hard.
}

For this section, let $\curveA$ and $\curveB$ be polygonal curves in $\Re^1$ unless otherwise stated, and let $\delta=1$ be the given threshold with no loss to generality. Since $\curveA$ and $\curveB$ are curves in $\Re^1$, we directly label column $i$ (resp.\ row $j$) of the free space with $\curveA_i$ (resp.\ $\curveB_j$). When modifications are restricted to one curve, they will be on $\curveB$, which then becomes $\curveB'$. 
We also define an arbitrary 3SAT instance as $\SATI$, with $\SATc$ clauses and $\SATv$ variables.

\subsection{Abstract Framework}
\seclab{abstract}

\RegVer{In this section we first describe the free space for the weak discrete \Frechet distance and how deletion or insertion can be used to close or create gaps in free paths. Then we give an abstract framework for our NP-hardness reductions, which in subsequent sections we will tailor to each specific problem.}

\subsubsection*{Paths and Gaps}

Recall from \secref{prelim} that for $\distDFrW{\curveA}{\curveB}$ the free space is an $m\times n$ grid graph, where vertex $(i,j)$ and vertex $(i',j')$ are adjacent if and only if $|i-i'|\leq 1$ and $|j-j'|\leq 1$. Then determining if $\distDFrW{\curveA}{\curveB} \leq 1$ is equivalent to determining if a path exists from $(1,1)$ to $(m,n)$ in the free space graph which only uses free vertices, namely vertices $(i,j)$ such that  $|\curveA_i-\curveB_j| \leq 1$. See \figref{fig:SimplePath}, for an example when such a path exists.
\RegVer{On the other hand, \figref{fig:SimpleDeletion} and \figref{fig:SimpleInsertion} show examples when no such path exists and thus $\distDFrW{\curveA}{\curveB} > 1$.}

Consider the highlighted pair of free vertices in \figref{fig:SimpleDeletion}. While their horizontal distance is 1, their vertical distance is 2, which we will refer to as a vertical \emph{gap} as it prevents a path through these vertices. Observe, however, that a deletion of the third vertex from $\curveB$ (i.e.\ the third row) removes this gap, creating a path from the lower left corner to the upper right corner, and thus $\dedistDFrW{\curveA}{\curveB}=1$. Conversely, observe that if we were only allowed insertions on $\curveB$, then there is no way to bridge this vertical gap. 
Now consider the highlighted pair of free vertices in \figref{fig:SimpleInsertion}, where now instead there is a horizontal gap. If we are only allowed deletions on $\curveB$ then there is no way to bridge this gap (though deletions on $\curveA$ would bridge the gap).  However, if we allow insertions on $\curveB$, then inserting a value of 20 at the third row would create a path between these two vertices, showing that $\iedistDFrW{\curveA}{\curveB}=1$. 
Thus in summary, deletion could be used to bridge a vertical gap but not a horizontal one, and insertion could be used to bridge a horizontal gap but not a vertical one. 

\begin{figure}[!h]
    \centering
    \begin{subfigure}[b]{0.3\textwidth}
        \centering
        \includegraphics[width=\textwidth]{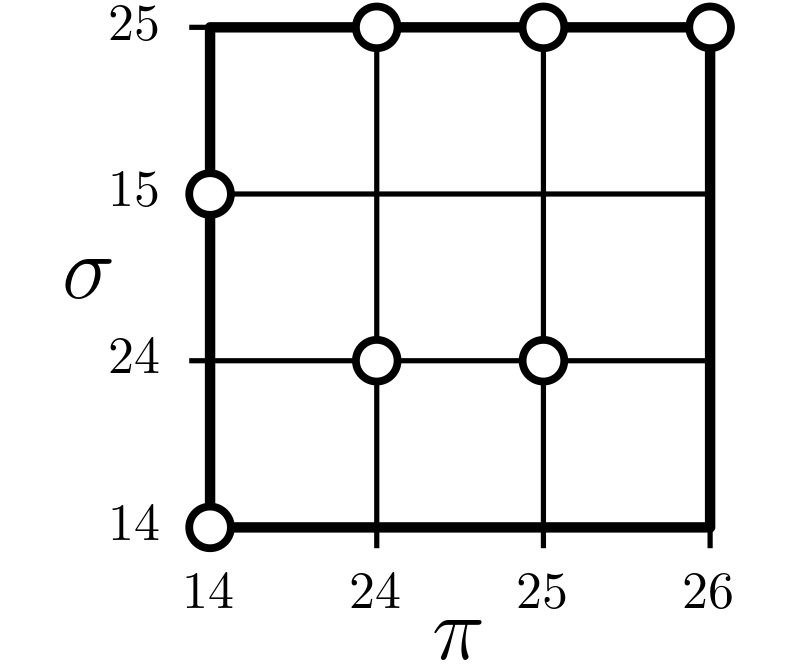}
        \subcaption{\centering The Weak \Frechet distance is already 1}
        \figlab{fig:SimplePath}
    \end{subfigure}
    \label{fig:SimpleFreeSpace}
    \hspace{1em}
    \begin{subfigure}[b]{0.3\textwidth}
        \centering
        \includegraphics[width=\textwidth]{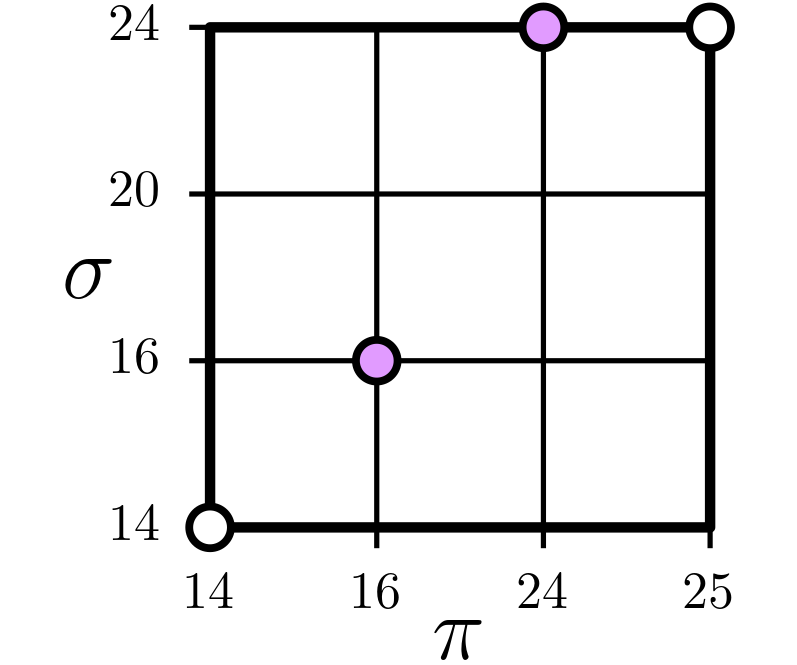}
        \subcaption{\centering Deletion can close the vertical gap.}
        \figlab{fig:SimpleDeletion}
    \end{subfigure}
    \hspace{1em}
    \begin{subfigure}[b]{0.3\textwidth}
        \centering
        \includegraphics[width=\textwidth]{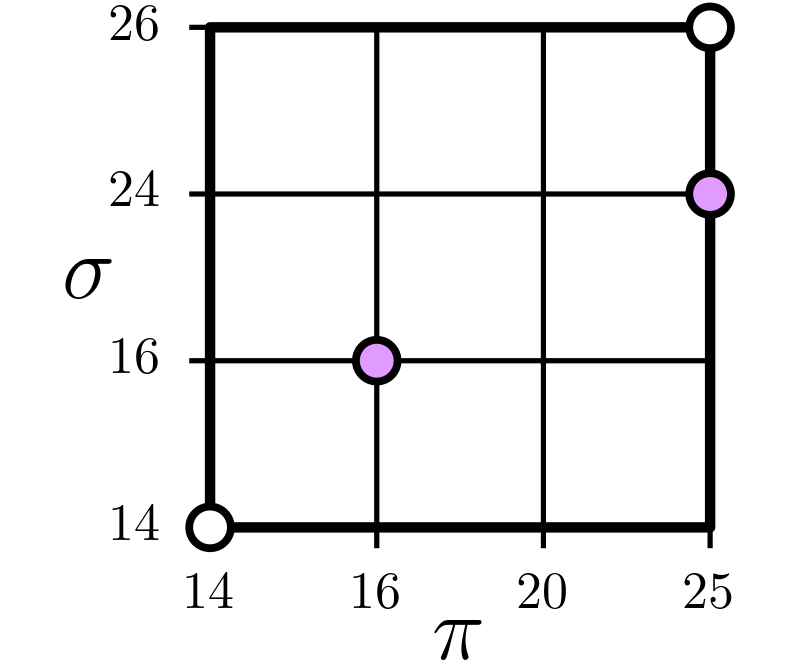}
        \subcaption{\centering Insertion can close the horizontal gap.}
        \figlab{fig:SimpleInsertion}
    \end{subfigure}
    \caption{Three free spaces diagrams with free spaces represented by circles, and edits permitted on $\curveB$ only. The values along the axes are the curve coordinates in $\Re^1$.}
    \figlab{fig:SimpleFreeSpace}
\end{figure}

Consider the example in  \figref{fig:OpposingDeletion},  where 
there are two vertical gaps, and suppose we are considering the deletion only problem. 
Now the first vertical gap can be removed by deleting $\curveB_3=16$. However, doing this creates a horizontal gap at $\curveA_8=16$, where the other vertical gap was, and this horizontal gap cannot be bridged by deletion(s). 
Similarly, if we start by trying to close the second vertical gap with deletion of row $\curveB_2=14$ then we get an insurmountable horizontal gap at $\curveA_2=14$. We thus refer to such a pair of vertical gaps as \emph{opposing}. 
Ultimately, our goal is to use the decision to create a path by bridging a gap as deciding to set a literal in the given 3SAT instance to \TRUE. Intuitively, by creating such opposing gaps we can make setting a literal to \TRUE correspond to setting instances of the negated literal to \FALSE. 

\begin{figure}[!h]
    \centering
    \includegraphics[width=0.4\textwidth]{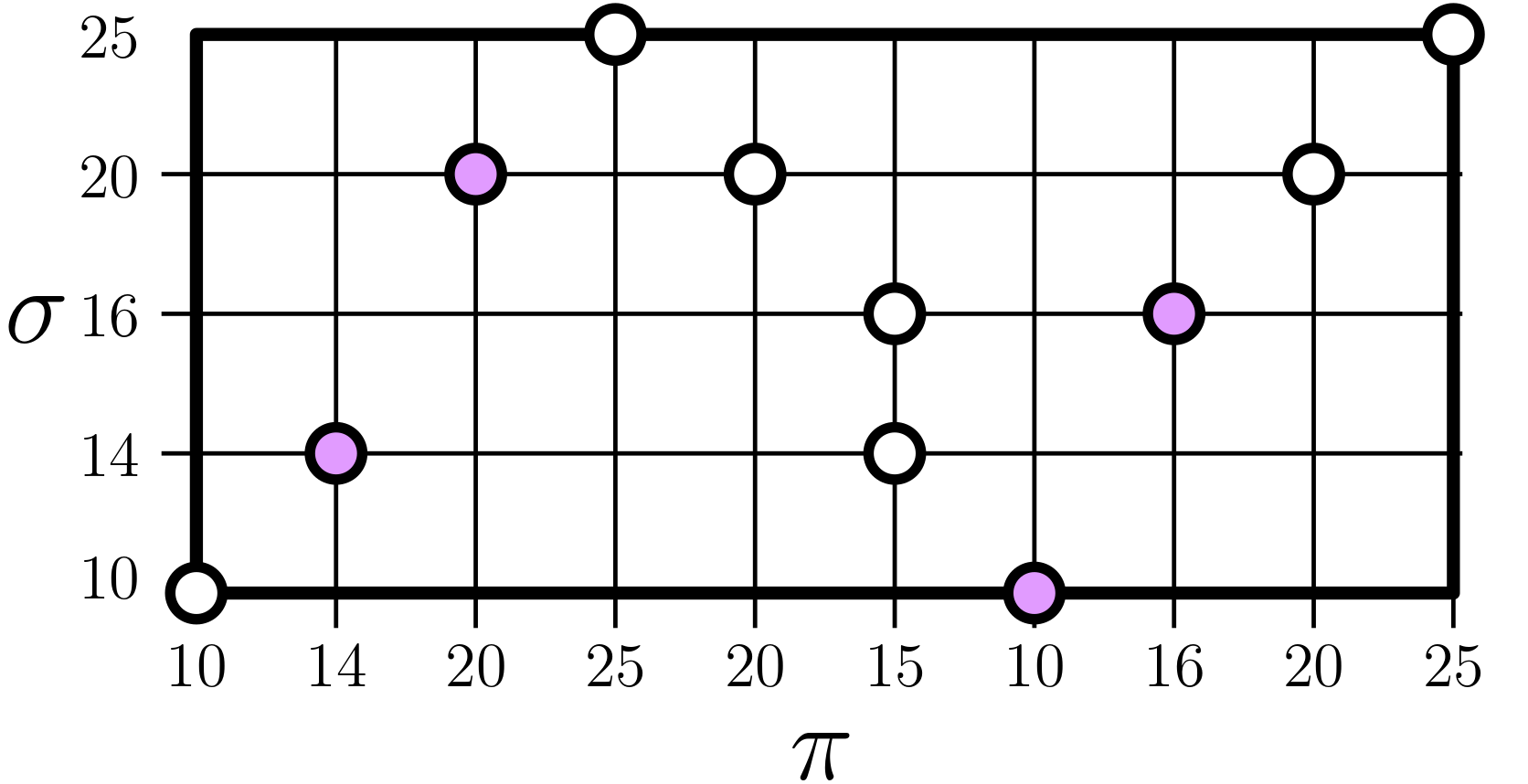}
    \caption{Opposing vertical gaps, where bridging one with a row deletion creates a horizontal gap at the other.}
    \figlab{fig:OpposingDeletion}
\end{figure}

\subsubsection*{Reduction Framework}

We now describe the abstract structure, shown in \figref{fig:SAT}, which we use to represent any 3SAT instance as an instance of weak discrete edit \Frechet  distance. 
First, we make a rectangular free space gadget for each clause, which are then placed in series. 
Within a given clause gadget, the rows can intuitively be partitioned into three layers, and each layer can be partitioned into three sections of columns. Thus overall the clause gadget consists of $9$ logical (roughly square) regions, 
where, as shown in \figref{fig:SAT}, each region consists of an orange diagonal path of free vertices, which we simply refer to as a \emph{diagonal}.  Now for the top and bottom layers, their three diagonals will be unobstructed and connect to each other to allow traversal through these regions. The middle layer will also consist of three diagonals, however, we create gaps on these diagonals to encode the given clause. Namely, the three diagonals will correspond to the three literals of the clause, and choosing to bridge a gap on one of these diagonals will correspond to setting that literal to $\TRUE$.
\RegVer{How one can enforce a correspondence between closing gaps and setting literals depends on which edit operations we are allowing, and the precise details are left to the relevant subsequent sections. For now, we simply claim that because we placed the clause gadgets in series, we will be able to enforce that closing a gap for a literal in one clause will close the gap for that literal across all clauses, while simultaneously creating an insurmountable gap for all instances of the negated literal (by using opposing gaps as described above), corresponding to setting the negated literal to \FALSE.}

\begin{figure}[!h]
    \centering
    \includegraphics[width=1\textwidth]{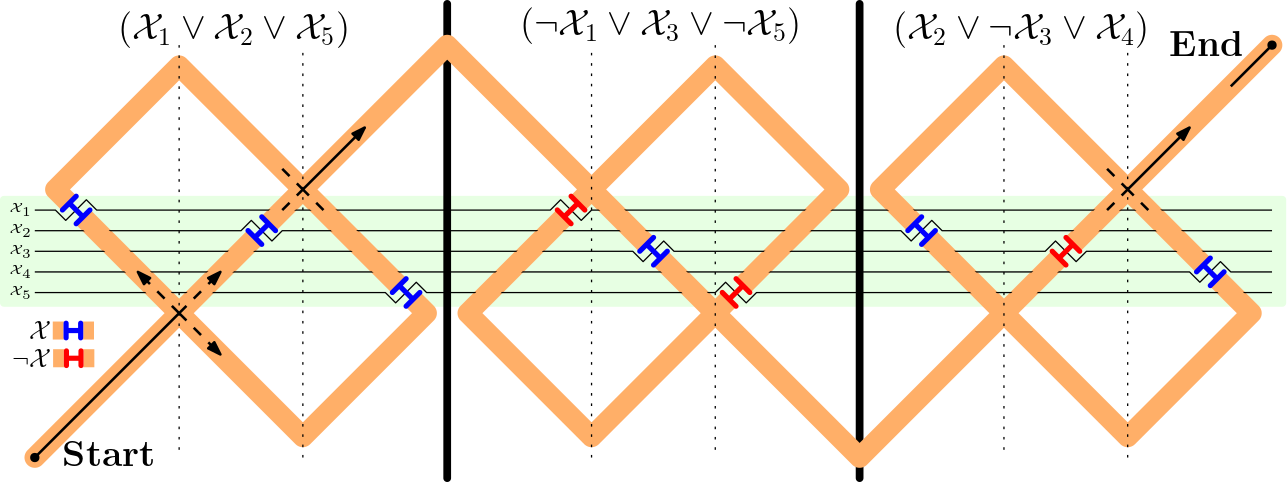}
    \caption{Abstract free space structure. This example is satisfied by setting $\BoolVar_2, \lnot \BoolVar_5 =\TRUE$.}
    \figlab{fig:SAT}
\end{figure}

As mentioned above, the clause gadgets are placed in series. Observe that we enter the first clause gadget at its bottom left corner, and exit at its top right corner. Thus in order to have the second clause gadget start where the first clause gadget ends, we invert the second clause gadget so that it must be traversed from its upper left corner to its lower right corner. In general, the odd clause gadgets must be traversed up and to the right, and the even ones down and to the right. (If there are an even number of clauses, we can insert one more gadget at the end that allows traversing from the lower left to the upper right.) Furthermore, when going from an odd to an even gadget, there will be a column inbetween with only a single free space at the top row (resp.\ bottom row when going from even to odd), to ensure this is the only point of connection between the gadgets.

\RegVer{With this abstract description, it is now easy to see the correspondence between the 3SAT instance and the weak discrete edit \Frechet distance. Namely, a solution to the 3SAT instance requires that we set a literal of every clause to \TRUE, and similarly there is only a path in the free space from $(1,1)$ to $(m,n)$ if we can edit $\curveB$ in such a way as to bridge at least one of the three diagonal gaps in the variable layer for every clause.
\subsubsection*{Building and Connecting Diagonals}
Here we describe how to select values to create and connect the diagonals shown in \figref{fig:SAT} and discussed in our framework above.
}

Let $L:=\seq{15, 25, \ldots 10v+5}$
be an ordered sequence of values,
and let $L^R$ denote $L$ in reverse order.
An ascending diagonal path is realized by setting portions of $\curveA$ and $\curveB$ to both $L$ or both $L^R$. Similarly, a descending diagonal path is created by setting a portion of $\curveA$ to $L$ (resp.\ $L^R$) and a portion of $\curveB$ to $L^R$ (resp.\ $L$). 

Consider a clause gadget, which consists of 9 diagonals, 3 in each layer, alternating between ascending and descending, creating a zigzag pattern. 
Within a layer, when two diagonals meet we insert a value in between them such that they are ``glued" together by a column which locally has no free vertices except at the one location where the diagonals come together. If the end of one diagonal (correspondingly the beginning of the next one) is the value $10v+5$, then this can be achieved by  placing the value $10(v+1)$ between the diagonals, as it is larger than any value in $L$. Similarly if the diagonal ends at the value $15$ then we insert the value $10$ before the next diagonal. These inserted values will also similarly act to glue the layers of the clause gadget together.
Let $\curveA^i$ denote the portion of $\curveA$ corresponding to the $i$th clause.
Then the \emph{basic clause gadget}, shown in \figref{fig:L's}, is defined by setting 
\[
\curveA^i = \curveB = \langle 0,10\rangle \circ L \circ \langle 10(v+1) \rangle \circ L^R \circ \langle 10 \rangle \circ L \circ \langle 10(v+1),10(v+2) \rangle
.\]

\begin{figure}[!h]
    \centering
    \includegraphics[width=0.5\textwidth]{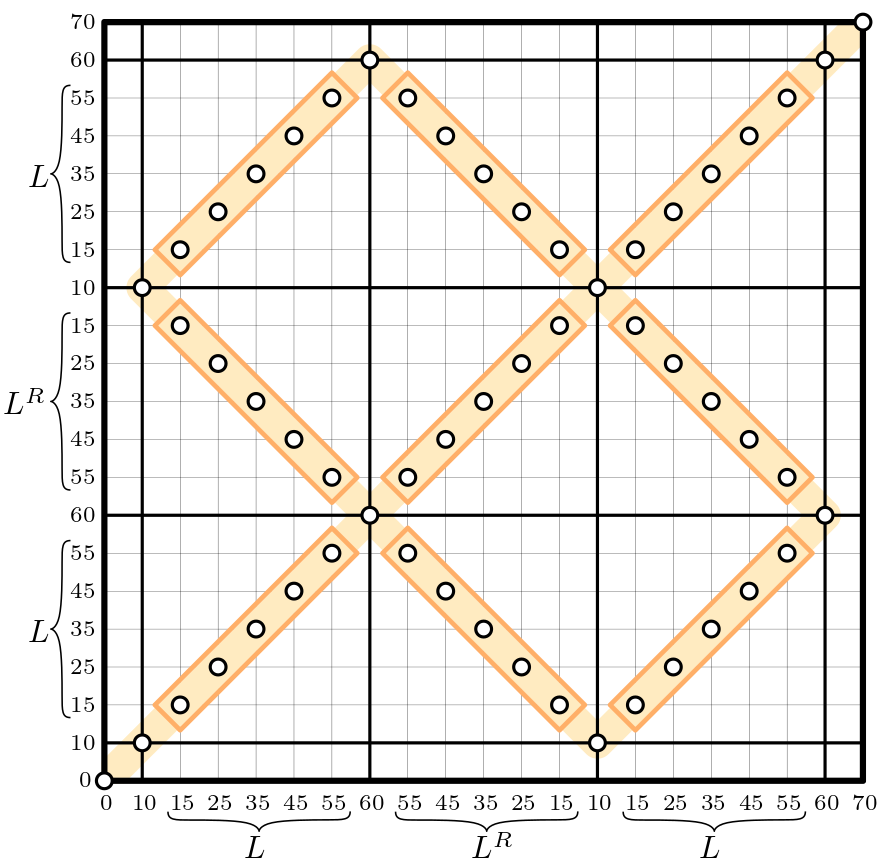}
    \caption {Basic clause gadget, consisting of 9 (highlighted) diagonals made by pairs of $L$'s and $L^R$'s which have been glued together such that the free-space has 3 paths.}
    \figlab{fig:L's}
\end{figure}
Observe that we appended the value $0$ at the beginning and $10(v+2)$ at the end of each curve. This serves to glue the successive clause gadgets together at single free vertices, in the same way we glued diagonals within a clause gadget together. Again, the values $0$ and $10(v+2)$ achieve this by respectively being smaller or larger than any value used internally in the clause gadget. 
Note that the above values used to create the basic clause gadget do not create any gaps in the variable layer.
Depending on the edit operation allowed, we modify the construction to create the appropriate gaps.
\SOCGVer{The details for these modifications are given in the full version.
We summarize the full suite of results in the following two theorems.

\begin{theorem}\thmlab{hardnessdiscretesocg}
    Given a value $\delta$ and curves $\curveA$ and $\curveB$ in $\Re^1$, determining if the weak discrete \Frechet distance between the curves can be made less than or equal to $\delta$ by any of the following processes is NP-hard:
    \RegVer{\begin{enumerate}[(a), leftmargin=1.5cm]}
    \SOCGVer{\begin{inparaenum}[(a)]}
        \item deleting any number of points from $\curveB$;
        \item deleting up to $k$ points from $\curveA$, $\curveB$, or both;
        \item inserting up to $k$ points into $\curveB$; and
        \item performing up to $k$ deletions or insertions from/into $\curveB$.
    \RegVer{\end{enumerate}}
    \SOCGVer{\end{inparaenum}}
    Further, determining if the weak discrete vertex-restricted shortcut \Frechet distance is less than or equal to $\delta$ is NP-hard.
\end{theorem}

\begin{theorem}\thmlab{hardnesscontinuoussocg}
    Given a value $\delta$ and curves $\curveA$ and $\curveB$ in $\Re^2$, determining if the weak continuous \Frechet distance between the curves can be made less than or equal to $\delta$ by any of the following processes is NP-hard:
    \RegVer{\begin{enumerate}[(a), leftmargin=1.5cm]}
    \SOCGVer{\begin{inparaenum}[(a)]}
        \item deleting any number of points from $\curveB$;
        \item deleting up to $k$ points from $\curveA$, $\curveB$, or both;
        \item inserting up to $k$ points into $\curveB$; and
        \item performing up to $k$ deletions or insertions from/into $\curveB$.
    \RegVer{\end{enumerate}}
    \SOCGVer{\end{inparaenum}}
    Further, determining if the weak continuous vertex-restricted shortcut \Frechet distance is less than or equal to $\delta$ is NP-hard.
\end{theorem}

}

\newcommand{\hardnessSubsection}{
\subsection{Weak Deletion Only}
Consider an instance $\SATI$ of 3SAT with $c$ clauses and $v$ variables. To prove that $\dedistDFrW{\curveA}{\curveB}$ is NP-complete we will select values for the points in $\curveA$ and $\curveB$, such that determining if a given number of row deletions in the free space (i.e. deletions from $\curveB$) will result in a path from  $(1,1)$ to $(m,n)$ equates to determining if there is a satisfying assignment for $\SATI$.
To do so we follow the abstract framework described above and shown in \figref{fig:SAT}, which we already know equates paths with satisfying assignments, but where now gaps will be implemented appropriately to work for the deletion only case. 

Consider the $i$th clause of $I$. Recall that in our definition of the basic clause gadget the rows of the middle variable layer were given by $L^R$, where 
$L$ is the ordered set of values $10i+5$ for all $1\leq i\leq v$.
For the deletion only case, we will replace this occurrence of $L^R$ with $\widehat{L}^R$, where $\widehat{L}$ is obtained from $L$ by replacing the value $10j+5$ with the two values $10j+4$ and $10j+6$, for all $1\leq j\leq v$. This elongates all diagonals in the variable layer, as it effectively creates a row not just for every variable, but for every possible assignment of every variable.

We also make substitutions for occurrences of $L$ and $L^R$ in the columns. Specifically, $L_k^+$ (resp.\ $L_k^-$) is obtained from $L$ by replacing the value $10k+5$ with $10k+6$ (resp.\ $10k+4$). Now let $\BoolVar_{k_1}$, $\BoolVar_{k_2}$, and $\BoolVar_{k_3}$ be the three variables that occur in the $i$th clause.
Then in the definition of $\curveA^i$ copied above, replace the first occurrence of $L$ with $L_{k_1}^{\pm}$, replace $L^R$ with $L_{k_2}^{\pm}$, and replace the second occurrence of $L$ with $L_{k_3}^{\pm}$, where $L_{k_i}^{\pm}=L_{k_i}^+$ if $\BoolVar_{k_i}$ appears as a positive literal and $L_{k_i}^{\pm}=L_{k_i}^-$ if $\BoolVar_{k_i}$ appears as a negated literal.

Observe that since we are using $\widehat{L}^R$ for the rows of the variable layer, any clause gadget containing $L_{k}^+$ in the columns, i.e.\ a clause with the positive literal $\BoolVar_{k}$, will have a diagonal in the variable layer, with a single gap at the row with value $10k+4$, as shown in \figref{fig:DeletionReduction}. Thus deleting this row will close this gap simultaneously for all clause gadgets containing $L_{k}^+$, allowing the diagonal to be traversed and intuitively setting $\BoolVar_{k}=\TRUE$. 
Similarly, clause gadgets containing $L_{k}^-$ in the columns, i.e.\ a clause with the literal $\lnot \BoolVar_{k}$, will have a gap at the row with value $10k+6$, and  deleting this row intuitively sets $\lnot \BoolVar_{k}=\TRUE$. However, by design any pair of such $L_{k}^+$ and $L_{k}^-$ induced gaps are opposing, as described above. 
Specifically, deleting the row with value $10k+4$ creates an horizontal gap at the row with value $10k+6$ for any clause gadget containing $L_{k}^-$, and this horizontal gap cannot be bridged by deleting rows from the variable layer. In other words, bridging a gap that corresponded to setting $\BoolVar_{k}=\TRUE$ prevents us from bridging any later gap corresponding to setting $\lnot \BoolVar_{k}=\TRUE$.

\begin{figure}[!h]
    \centering
    \includegraphics[width=1\textwidth]{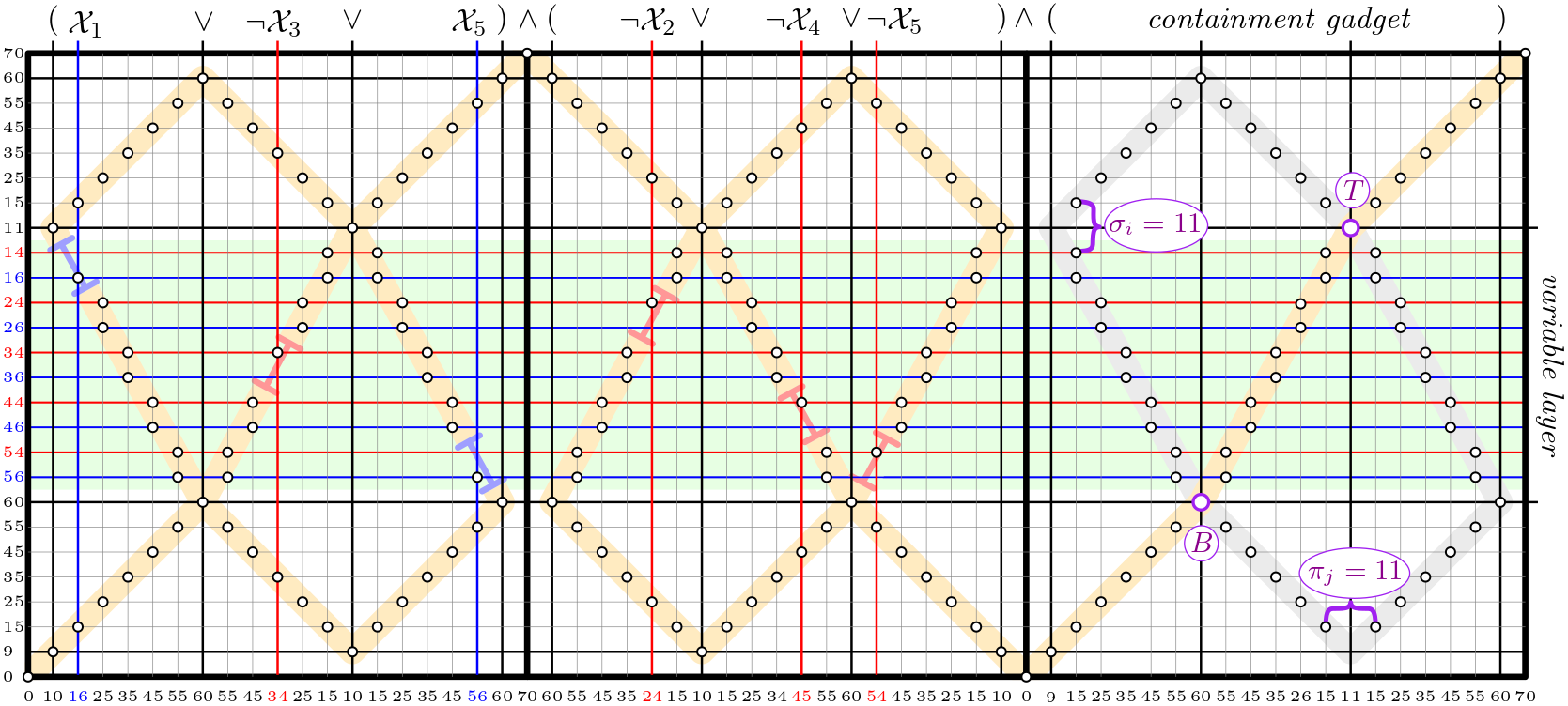}
    \caption{Free space example for $\dedistDFrW{\curveA}{\curveB}$ reduction.
    Observe that deletion of row 54 closes the vertical gap for $\BoolVar_5$, but creates a horizontal gap for $\lnot\BoolVar_5$, i.e.\ setting $\BoolVar_5$ to \TRUE sets $\lnot \BoolVar_5$ to \FALSE (and vice versa for deleting 56). 
    The containment gadget restricts deletion to the variable layer only.}
    \figlab{fig:DeletionReduction}
\end{figure}

Given the above discussion, both in this and the prior section, it is easy to see that there is a solution to the given 3SAT instance $\SATI$ if and only if $\dedistDFrW{\curveA}{\curveB}\leq 1$, when deletions are restricted to the variable layer. As allowing deletions outside the variable layer may break this correspondence,\footnote{For example, in \figref{fig:DeletionReduction} one could delete the entire variable layer and the entire top layer except for the top row.  
Doing so will yield $\distDFrW{\curveA}{\curveB'}\leq 1$ for any instance $\SATI$.}
we add one final containment gadget at the end of $\curveA$, as shown in \figref{fig:DeletionReduction}, which effectively restricts deletion to only the variable layer. 
Let $\overline{\curveA}$ denote the portion of $\curveA$ corresponding to this final containment gadget.
We start by setting $\overline{\curveA}=\curveA^i$, where $\curveA^i$ is from the basic clause gadget as copied above.
Now, we simply replace the first \(10\) with \(9\) and the second \(10\) with \(11\). To complete the containment gadget, we also need to modify $\curveB$, which is a basic clause gadget, except where $L^R$ is replaced with $\widehat{L}^R$ as discussed above. We now further modify $\curveB$ by also replacing the first $10$ with $9$ and the second $10$ with $11$ (i.e.\ the same modifications used for $\overline{\curveA}$).

Let $S$ denote the lower left starting free vertex of the containment gadget, and let $E$ be the upper right ending free vertex.
Let $T$ denote the free vertex at $\curveA_j=\curveB_i=11$, and let $B$ denote the lower left free vertex such that $\curveA_k=\curveB_\ell=10(v+1)$, see \figref{fig:DeletionReduction}. 
We now argue that the row containing $T$ (resp. $B$) cannot be deleted, which in turn implies that no row between $T$ and $E$ (resp. $S$ and $B$) can be deleted, as this would create an insurmountable horizontal gap on the diagonal between $T$ and $E$ (resp. $S$ and $B$), which is the only viable path from $T$ to $E$ (resp. $S$ to $B$).
As shown in \figref{fig:DeletionReduction}, $T$ is the only free vertex in its column, thus clearly its row cannot be deleted as it would create an insurmountable horizontal gap. 
Moreover, $T$ is the only free vertex in its row, thus the path from $S$ to $T$ must be confined to the rows below $T$. 
However, $B$ is the only free vertex in its column when restricting to the subset of row below $T$, and thus again deleting $B$ would create an insurmountable horizontal gap.

Finally, as this containment gadget was added in series, it must be traversed to satisfy $\dedistDFrW{\curveA}{\curveB} \leq 1$. 
Thus the containment clause effectively restricted deletions to the variable layer, and observe it did so without having to put a bound on the number of allowed deletions.

Above we described and argued correctness for the reduction, but here we give a final compact description for quick reference.
We are given an instance $\SATI$ of 3SAT, with $\SATv$ variables and $\SATc$ clauses. 
Let $L$ be the ordered set of the elements $10i + 5$ for all $1 \leq i \leq \SATv$. 
Let $\widehat{L}$ be obtained from $L$ by replacing the value $10j+5$ with the two values $10j+4$ and $10j+6$, for all $1\leq j\leq v$.
Let $L_k^+$ (resp.\ $L_k^-$) be obtained from $L$ by replacing the value $10k+5$ with $10k+6$ (resp.\ $10k+4$). Finally let $S^R$ denote any ordered set $S$ in reverse order. Then we have the following construction of $\curveA$ and $\curveB$.

\begin{itemize}
    \item Let $\curveA^i$ represent clause $i$ of $\SATI$ which contains variables $\BoolVar_{k_1}$, $\BoolVar_{k_2}$, and $\BoolVar_{k_3}$ and therefore $\curveA^i = \langle 10 \rangle \circ L_{k_1}^{\pm} \circ \langle 10(\SATv+1) \rangle \circ (L_{k_2}^{\pm})^R \circ \langle 10 \rangle \circ L_{k_3}^{\pm} \circ \langle 10(\SATv+1) \rangle$, where $L_{k_i}^{\pm}=L_{k_i}^+$ if $\BoolVar_{k_i}$ appears as a positive literal and $L_{k_i}^{\pm}=L_{k_i}^-$ if $\BoolVar_{k_i}$ appears as a negated literal.
    \item Let $\overline{\curveA}$ represent the containment gadget and be $\langle 9 \rangle \circ L \circ \langle 10(\SATv+1) \rangle \circ L^R \circ \langle 11 \rangle \circ L \circ \langle 10(\SATv+1) \rangle$
    \item Let $\curveA=\langle 0 \rangle \circ \curveA^1 \circ \langle 10(\SATv+2) \rangle \circ (\curveA^2)^R \circ \langle 0 \rangle \circ \dots \circ  (\curveA^\SATc)^R \circ \langle 0 \rangle \circ \overline{\curveA} \circ \langle 10(\SATv+2) \rangle$ (if $\SATc$ is odd, duplicate one clause so the total number of clauses is odd).
    \item Let $\curveB=\langle 0,9 \rangle \circ L \circ \langle 10(\SATv+1) \rangle \circ \widehat{L}^R \circ \langle 11 \rangle \circ L \circ \langle 10(\SATv+1),10(\SATv+2) \rangle$.
\end{itemize}

\begin{theorem}\thmlab{deletionmain}
    Given a value $\delta$ and curves $\curveA$ and $\curveB$ in $\Re^1$, determining if the weak discrete \Frechet distance between the curves can be made less than or equal to $\delta$ by deleting any number of points from $\curveB$, is NP-hard.
\end{theorem}

If we restrict the number of deletions to $\SATv$ (the most possibly required to properly `assign' variables), we can remove the containment gadget, and instead prevent deletion of non-variable rows by simply duplicating them $\SATv$ times.\footnote{It is not enough to restrict the deletions alone. Observe in \figref{fig:DeletionReduction} that deleting rows with values 11, 14, and 16 would allow traversal of the clause gadgets even in an unsatisfiable case.
In particular, the gadgets for unsatisfiable formula $(\BoolVar_1 \lor \BoolVar_1 \lor \BoolVar_1)\land(\lnot\BoolVar_1 \lor \lnot\BoolVar_1 \lor \lnot\BoolVar_1)\land(\BoolVar_2 \lor \BoolVar_3 \lor \BoolVar_4)$ can be traversed by deleting said rows and row 44 for $\BoolVar_4 =$ \TRUE.} 
This works as duplicating a row does not affect connectivity, and our total deletion budget is not enough to remove all copies of a duplicated row. 
The same trick can be done for all columns in $\curveA$, effectively preventing their deletion even if we allow $\SATv$ deletions from both curves.

\begin{corollary}\corlab{uptok}
    Given values $\delta$ and $k$ and curves $\curveA$ and $\curveB$ in $\Re^1$, determining if the weak discrete \Frechet distance between the curves can be made less than or equal to $\delta$ by deleting up to $k$ points from $\curveA$, $\curveB$, or both, is NP-hard.
\end{corollary}

We can extend the results above to the continuous case $\dedistFrW{\curveA}{\curveB}$ by lifting the curves to $\Re^2$.\footnote{Without going to $\Re^2$, gaps like in \figref{fig:SimpleDeletion} could be passed without deletion, yielding $\distFrW{\curveA}{\curveB} \leq 1$.} To do this, we change every existing value $i$ to $(i,0)$ and insert $(0,\infty)$ between every point in $\curveA$ and every point in $\curveB$ 
(if not using the containment gadget, these inserted points should also be duplicated $\SATv$ times to prevent their deletion). This forces movement along the curves to mimic discrete movement as is illustrated in \figref{fig:ContinuousJustification}.
\begin{figure}[t]
    \centering
    \begin{subfigure}[b]{1\textwidth}
        \centering
        \includegraphics[width=\textwidth]{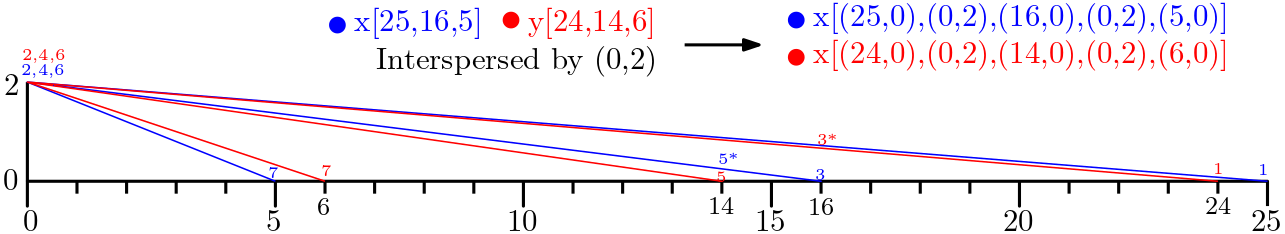}
        \subcaption{\centering Continuous distance is $<=1$ when discrete is $>1$ after interspersing by $(0,2)$. Numbers represent simultaneous traversal order and $i^*$ indicates an invalid stop were it discrete \Frechet.}
        \figlab{fig:ContinuousTwo}
    \end{subfigure}
    \label{fig:ContinuousTwo}
    \vspace{2em}
    \begin{subfigure}[b]{1\textwidth}
        \centering
        \includegraphics[width=\textwidth]{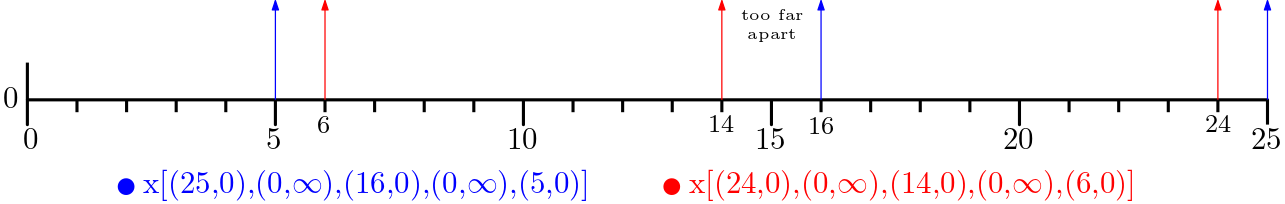}
        \subcaption{\centering Continuous distance is the same as discrete after interspersing by $(0,\infty)$.}
        \figlab{fig:ContinuousInfty}
    \end{subfigure}
    \caption{Red and blue represent the different curves which are shown in their actual $\Re^2$ image space, i.e.\ these figures are no showing the free space.}
    \figlab{fig:ContinuousJustification}
\end{figure}
Note that the insertion of these vertices does not affect the discrete distance $\dedistDFrW{\curveA}{\curveB}$. Specifically, suppose on the original curves we performed the move $(\curveA_i, \curveB_j)\rightarrow (\curveA_{i\pm 1},\curveB_{j\pm 1})$, then on the new curves at the same cost we can perform the move $(\curveA_i, \curveB_j)\rightarrow ((0,\infty),(0,\infty))\rightarrow (\curveA_{i\pm 1},\curveB_{j\pm 1})$. Similarly the move $(\curveA_i, \curveB_j)\rightarrow (\curveA_{i\pm 1},\curveB_{j})$, would become $(\curveA_i, \curveB_j)\rightarrow ((0,\infty),(0,\infty))\rightarrow (\curveA_{i\pm 1},\curveB_{j})$. Moreover, as we always must simultaneously move to $(0,\infty)$, note that $\dedistDFrW{\curveA}{\curveB}$ could not have decreased.
Now switching from discrete to continuous \Frechet, 
recall that the weak continuous \Frechet distance is realized either at a vertex to vertex distance or a vertex to edge distance. If in our construction it was realized at a vertex to vertex distance then it is equivalent to the discrete case. So suppose it was realized at a vertex to edge distance, say vertex $\curveA_i$ and edge $\curveB_j(0,\infty)$. However, observe that $\curveB_j(0,\infty)$ is a vertical edges as shown in \figref{fig:ContinuousJustification}. Thus as $\curveA_i$ and $\curveB_j$ are both on the $x$-axis, the closest to point $\curveA_i$ on the edge $\curveB_j(0,\infty)$ is $\curveB_j$, i.e.\ a vertex to vertex distance.\footnote{Rather than using $\infty$, it would suffice to use a sufficiently large finite value $x$. The projection onto a now near-vertical segment would lie slightly off the $x$-axis, potentially slightly lowering the cost. However, this won't matter so long as $x$ is sufficiently large, since we are using integer coordinates and $\delta=1$.}


\begin{corollary}\corlab{deletioncont}
    Given a value $\delta$ and curves $\curveA$ and $\curveB$ in $\Re^2$, determining if the weak continuous \Frechet distance between the curves can be made less than or equal to $\delta$ by deleting any number of points from $\curveB$, is NP-hard.
    
    It is also NP-hard if, given an additional value $k$, deletions are limited to $k$ deletions from $\curveA$, $\curveB$, or both.
\end{corollary}

As a final result, we now observe how the above extends to the weak vertex-restricted shortcut problem, similarly to how our results in \secref{cdo} extended to the strong vertex-restricted shortcut problem (studied in \cite{df-jdcfds-13,bds-cfdsn-14}).  
Recall that in this problem, $\curveA$ is fixed, and on $\curveB$ you are allowed to shortcut directly from $\curveB_i$ to $\curveB_j$ for any $i<j$ (i.e.\ replace the subcurve with the line segment between its endpoints). You are allowed to shortcut as often as you like and the question is whether you can get the (now weak) \Frechet distance between the resulting curves to be $\leq \thresh$. 
Observe, however, unlimited shortcutting is equivalent to unlimited deletion (i.e. the case considered in 
\thmref{deletionmain}), except deleting the starting or ending vertices must be prohibited as that cannot be achieved by shortcutting. However, we can just add this restriction to the above reduction and it still works. 
Thus we have the following result, analogous to that in \thmref{deletionmain} and the first part of \corref{deletioncont}.


\begin{corollary}
    Given a value $\delta$ and curves $\curveA$ and $\curveB$ in $\Re^1$, determining if the weak discrete vertex-restricted shortcut \Frechet distance is less than or equal to $\delta$ is NP-hard.
    Moreover, for curves $\curveA$ and $\curveB$ in $\Re^2$, determining if the weak continuous vertex-restricted shortcut \Frechet distance is less than or equal to $\delta$ is NP-hard.
\end{corollary}

\subsection{Weak Insertion}

We now describe how the reduction used in the prior section for deltion only, as formally described just before \thmref{deletionmain} and shown in \figref{fig:DeletionReduction}, can be modified for insertion only, i.e.\ $\iedistDFrW{\curveA}{\curveB}$ and $\iedistFrW{\curveA}{\curveB}$. 
At the end of this section, we remark how this construction can be easily modified to allow both insertions and deletions, i.e.\ $\edistDFrW{\curveA}{\curveB}$ and $\edistFrW{\curveA}{\curveB}$.

First, on $\curveB$ replace $\widehat{L}$ with $\emptyset$. 
This creates horizontal gaps of width $\SATv$ in what was the variable layer, which can only be overcome by inserting $\SATv$ rows. 
We wish for the choice of each of the values inserted to overcome these gaps
to correspond to the assignment of a variable. 
Now previously for deletion, assigning $\BoolVar_k$ to \TRUE corresponded to deleting row $10k+4$, and \FALSE to deleting row $10k+6$. 
The problem with insertion is that we are not limited to choosing $10k+4$ and $10k+6$, and could instead insert $10k+5$, which will satisfy columns for both $10k+4$ and $10k+6$, i.e.\ both  $\BoolVar_k$ and $\lnot\BoolVar_k$. 
To prevent inserting $10k+5$, we change the definition of $L_k^+$ (resp.\ $L_k^-$) to be obtained from $L$ by replacing the value $10k+5$ with $10k+7$ (resp.\ $10k+3$). 
This changes the values for columns in $\curveA$ representing $\BoolVar_k$ and $\lnot\BoolVar_k$, while the absence of a variable in a clause is still represented by a column of value $10k+5$. 
This restricts insertions to $10k+4$ and $10k+6$ since they are the only values that are close enough to $10k+5$, while also covering either $10k+3$ or $10k+7$, respectively.
Thus we now have a way to represent setting the variable $\BoolVar_k$. 

Unfortunately, the top and bottom layers of rows are now ruined as, by design, there is no row value that will allow passage for both $10k+3$ and $10k+7$ columns simultaneously. 
To fix this, we change both curves by replacing every one of the 9 diagonals (as defined in \secref{abstract}) of our new clause gadget with an appropriately modified basic clause gadget, creating a larger gadget as shown \figref{fig:InsertionReduction}.
%
Thus the original bottom layer, now has itself three sublayers of rows. 
We use the bottom sublayer to allow horizontal traversal over $\lnot \BoolVar_k$, and the middle sublayer for $\BoolVar_k$ (the top sublayer is not utilized in any particular way). 
In this way, regardless of whether $\BoolVar_k$ or $\lnot \BoolVar_k$ appears in the given clause, there will be a way to traverse the bottom layer. The sublayers of the top layer will be identical, to again allow traversal regardless of whether the clause contains $\BoolVar_k$ or $\lnot \BoolVar_k$.
The specific values used for these layers to achieve this are described below and shown in \figref{fig:InsertionReduction}.

\begin{figure}[!h]
    \centering
    \includegraphics[width=1\textwidth]{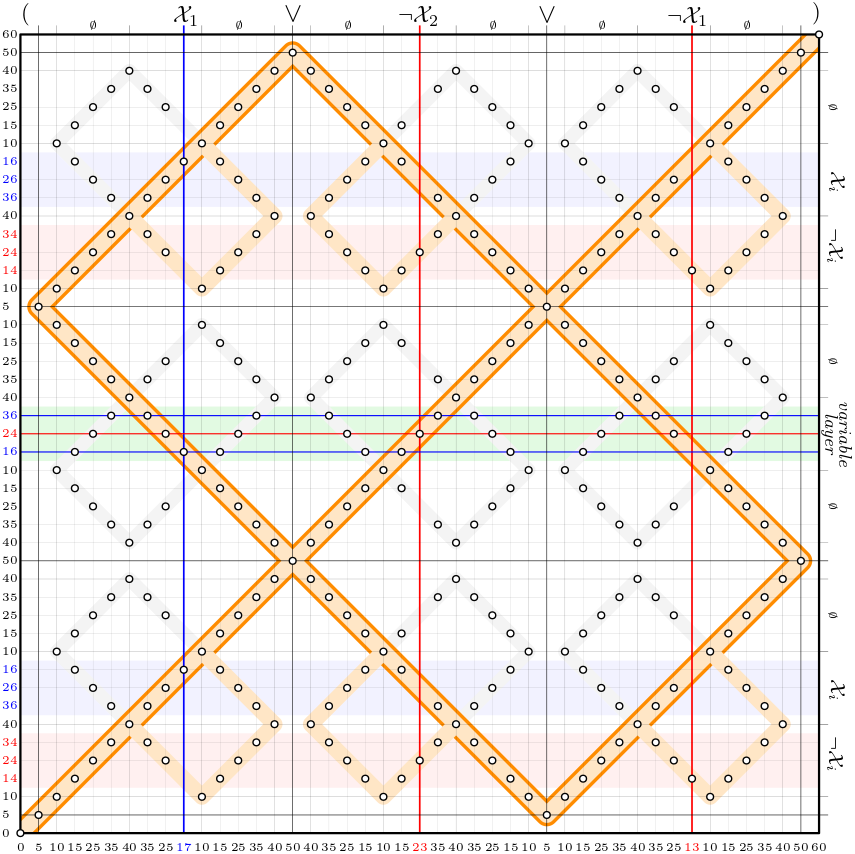}
    \caption{The new enlarged clause gadget for insertion only. 
    Specifically, the trivial clause $(\BoolVar_1 \lor \lnot \BoolVar_2 \lor \lnot \BoolVar_1)$ is represented, with the assumption that there is later an $\BoolVar_3$. Rows 16, 24, and 36 are already inserted. As before, the thick pale orange line shows paths that could be taken in some 3SAT instances, while grey paths are never useful. The dark orange border line merely highlights that the general structure seen in prior figures is still present. As indicated on the right, the pale blue and red strips indicate sublayers permitting travel across $\BoolVar$ and $\lnot\BoolVar$, respectively.}
    \figlab{fig:InsertionReduction}
\end{figure}

The middle row layer is now likewise divided into three row sublayers. Recall the columns in a clause gadget were also divided into three sections, one for each literal, each of which is now divided into three subsections.
Consider any column section corresponding to a literal $\BoolVar_k$ or $\lnot \BoolVar_k$.
By leaving the top and bottom sublayers of the middle row layer as $10i+5$ for all $i$, there will be a column of value $10k+7$ or $10k+3$ that prevents traversal in the middle column subsection of this literal in these sublayers due to a vertical gap that cannot be fixed with insertion. Therefore, the only way across the middle column subsections of the middle row layer will be through the middle row sublayer, i.e.\ our new variable sublayer. As discussed above in this section, we leave this sublayer empty, forcing $\SATv$ insertions to bridge this $\SATv$ width horizontal gap.
Again, setting the portions of $\curveA$ as described above for each middle subsection of each clause gadget, will imply the inserted rows will correspond to a satisfying assignment to $\SATI$ if they allow passage, and a non-satisfying assignment if they do not allow passage.
This correspondence holds so long as insertions are limited to this variable sublayer, which will be ensured by requiring a budget of exactly $\SATv$ insertions.
%
Specifically, the variable sublayer must clearly be crossed at least once in order to reach the end of $\curveB$, and this requires all $\SATv$ insertions as it contains a horizontal gap of width $\SATv$.
Thus there are no remaining insertions available to attempt any other insertions that might break the correspondence, such as attempting to span horizontal gaps between clauses.

Given an instance $\SATI$ of 3SAT, with $\SATv$ variables and $\SATc$ clauses, our precise construction is thus as follows. 
Let $L$ be the ordered set of the elements $10i + 5$ for all $1 \leq i \leq \SATv$.
Let $L^+$ (resp.\ $L^-$) be $L$ but with all values shifted up 1 (resp.\ down 1).
Let $L_k^+$ (resp.\ $L_k^-$) be $L$ but with the value $10k+5$ replaced with $10k+7$ (resp.\ $10k+3$).
We then define groupings of these ordered sets into larger ordered sets, letting $S^R$ denote any ordered set $S$ in reverse order.
Let $G$ be $\langle 10 \rangle \circ L^- \circ \langle 10(\SATv+1) \rangle \circ (L^+)^R \circ \langle 10 \rangle \circ L \circ \langle 10(\SATv+1)\rangle$.
Let $\widehat{G}$ be $\langle 10 \rangle \circ L \circ \langle 10(\SATv+1) \rangle \circ \emptyset \circ \langle 10 \rangle \circ L \circ \langle 10(\SATv+1)\rangle$
Let $G_{k_i}^{\pm}$ represent a literal and be $\langle 10 \rangle \circ L \circ \langle 10(\SATv+1) \rangle \circ (L_{k_i}^{\pm})^R \circ \langle 10 \rangle \circ L \circ \langle 10(\SATv+1) \rangle$, where 
$G_{k_i}^{\pm}=G_{k_i}^{+}$ and $L_{k_i}^{\pm}=L_{k_i}^{+}$ if 
the literal is $\BoolVar_{k_i}$, and $G_{k_i}^{\pm}=G_{k_i}^{-}$ and $L_{k_i}^{\pm}=L_{k_i}^{-}$ if the literal is $\lnot\BoolVar_{k_i}$.
Then we have the following construction of $\curveA$ and $\curveB$.
\begin{itemize}
    \item Let $\curveB = \langle 0,5 \rangle \circ G \circ \langle 10(\SATv+2) \rangle \circ \widehat{G}^R \circ \langle 5 \rangle \circ G \circ \langle 10(\SATv+2),10(\SATv+3) \rangle$.
    \item Let $\curveA^i$ represent clause $i$ of $\SATI$ which contains variables $\BoolVar_{k_1}$, $\BoolVar_{k_2}$, and $\BoolVar_{k_3}$ and therefore $\curveA^i = \langle 5 \rangle \circ G_{k_1}^{\pm} \circ \langle 10(\SATv+2) \rangle \circ (G_{k_2}^{\pm})^R \circ \langle 5 \rangle \circ G_{k_3}^{\pm} \circ \langle 10(\SATv+2) \rangle$.
    \item Let $\curveA = \langle 0 \rangle \circ \curveA^1 \circ \langle 10(\SATv+3) \rangle \circ (\curveA^2)^R \circ \langle 0 \rangle \circ \dots \circ \curveA^\SATc \circ \langle 10(\SATv+3) \rangle$ (if $\SATc$ is even, duplicate one clause so the total number of clauses is odd).
\end{itemize}

\begin{theorem}\thmlab{hardinsert}
    Given values $\delta$ and $k$ and curves $\curveA$ and $\curveB$ in $\Re^1$, determining if the weak discrete \Frechet distance between the curves can be made less than or equal to $\delta$ by inserting up to $k$ points into $\curveB$ is NP-hard.
\end{theorem}

Recall that for the analogous first theorem for the deletion only case, \thmref{deletionmain}, unlimited deletions were allowed (via the addition of a containment gadget).
We observe here, however, that unlimited insertions on $\curveB$ is easily polynomial time solvable.
Specifically, if there is a row in the free space with no free vertex, then this row is not passable regardless of what other rows are inserted. Conversely, two free vertices in consecutive rows can always reach each other by inserting a diagonal in between them (i.e.\ inserting rows corresponding to the column values in between them).
Thus the \Frechet distance can be made $\leq \delta$ if and only if every row of the free space has at least one free vertex, or stated in terms of the values on $\curveB$ and $\curveA$, every point in $\curveB$ is within $\delta$ of some point in $\curveA$.
Moreover, note that unlimited insertions on both curves always allows a \Frechet distance of 0 by concatenating $\curveA$ before $\curveB$ on $\curveB$ and $\curveB$ after $\curveA$ on $\curveA$. 

While unlimited insertions will not yield a hardness result, we can still modify the above reduction to extend to the continuous case. This is achieved in a similar way as was done for deletion, namely moving the curves to the $x$-axis in $\Re^2$ and then placing $(0,\infty)$ between consecutive vertices. However, now for our empty variable sublayer on $\curveB$, we instead include $\SATv-1$ such $(0,\infty)$ points, in preparation of variable assignment rows being inserted between them.\footnote{It is not possible to `cheat' this set-up by inserting gate-rows after (not in-between) the $(0,\infty)$ points, because on $\curveA$ the variables are in-between the $(0,\infty)$ points, which forces the same on $\curveB$ as seen in \figref{fig:SimpleInfinityInsertion}.}
\begin{figure}[t]
    \centering
    \includegraphics[width=0.5\textwidth]{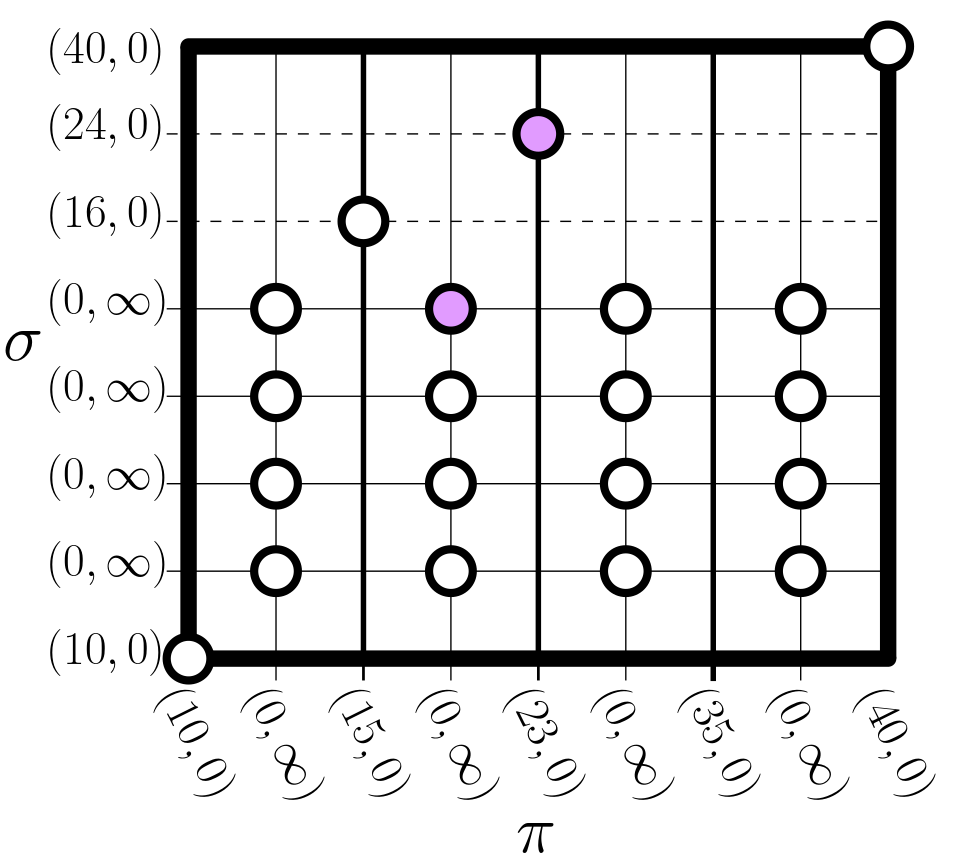}
    \caption{A zoomed in example of continuous insertion over $\lnot \BoolVar_2$ where $16$ and $24$ (dashed) have been inserted after the $(0,\infty)$ rows in an attempt to cheat the system. The created gap (highlighted) prevents this from working.}
    \figlab{fig:SimpleInfinityInsertion}
\end{figure}

\begin{corollary}\corlab{continousinserthard}
    Given values $\delta$ and $k$ and curves $\curveA$ and $\curveB$ in $\Re^2$, determining if the weak continuous \Frechet distance between the curves can be made less than or equal to $\delta$ by inserting up to $k$ points into $\curveB$, is NP-hard.
\end{corollary}



Finally, recall that for the deletion only case, by imposing a budget of $\SATv$ deletions, we were able to confine deletions to the variable layer by duplicating each row of the other layers $\SATv$ times. Thus duplicating all rows in the insertion only construction, similarly means that allowing budgeted deletion of the rows will not affect the \Frechet distance. We thus have the following result, which could be stated as a corollary, though we instead state as a theorem as it now applies to  
$\edistDFrW{\curveA}{\curveB}$ and $\edistFrW{\curveA}{\curveB}$ rather than $\iedistDFrW{\curveA}{\curveB}$ and $\iedistFrW{\curveA}{\curveB}$.

\begin{theorem}\thmlab{hardedit}
Given values $\delta$ and $k$ and curves $\curveA$ and $\curveB$ in $\Re^1$, determining if the $\delta$-threshold discrete \Frechet edit distance is less than or equal to $k$ is NP-hard.

Moreover, for curves $\curveA$ and $\curveB$ in $\Re^2$, determining if the $\delta$-threshold continuous \Frechet edit distance is less than or equal to $k$ is NP-hard.
\end{theorem}
}

\RegVer{\hardnessSubsection}



\bibliographystyle{plainurl}
\bibliography{refs}

\RegVer{
\appendix

\section{Minimum Vertex Curve Endpoints Reduction}\apndlab{mvcreduction}

Let $\mv{\curveA}$ be the analogue of $\mv{s,t,\curveA}$ from \defref{mv}, except where the starting points $s$ and $t$ are not specified. 
Here we show that the problem of computing $\mv{s,t,\curveA}$ can be reduced to computing $\mv{\curveA'}$, where $\curveA'$ is obtained from $\curveA$ by prepending and appending a constant number of vertices.
So consider any line $\ell$ through $s$. 
There are two points on $\ell$ at distance exactly $\thresh$ from $s$, call them $p_1$ and $p_2$. 
Analogously define the points $q_1$ and $q_2$ for $t$. Then we set $\curveA' = \seq{p_1,p_2,p_1,p_2,s}\circ \curveA \circ \seq{t,q_2,q_1,q_2,q_1}$ as shown in \figref{fig:MV_PiPrime}.

\begin{figure}[h]
    \centering
    \begin{subfigure}[b]{.45\textwidth}
        \centering
        \includegraphics[width=1\textwidth]{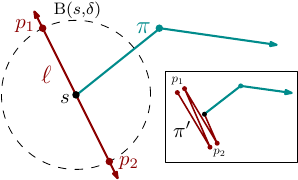}
        \subcaption{An example of $\ell$ and $\curveA$ being used to construct $\curveA'$ (space added along $\ell$ for readability)}
        \figlab{fig:MV_PiPrime}
    \end{subfigure}
    \hspace{2em}
    \begin{subfigure}[b]{.45\textwidth}
        \centering
        \includegraphics[width=1\textwidth]{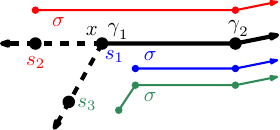}
        \subcaption{Three cases of $s$, where \textcolor{blue}{$s_1 \Rightarrow \curveB = \gamma$}, \textcolor{red}{$s_2 \Rightarrow \curveB=\seq{s,\gamma_2,\ldots,\gamma_k}$}, and \textcolor{seagreen}{$s_3 \Rightarrow \curveB=s\circ \gamma$}}
        \figlab{fig:MV_s_cases_complex}
    \end{subfigure}
    \caption{}
    \figlab{fig:MV_figs}
\end{figure}

Let $\mv{\curveA'}=\curveB'=\{\curveB_1',\ldots,\curveB_n'\}$.
Since $\distFr{\curveA'}{\curveB'}\leq \thresh$, there is a bijective mapping between traversals of the two curves such that paired points are within distance $\thresh$. 
Fix any such $\thresh$-realizing traversal. 
Let $x$ denote the point on $\curveB'$ which the point $s$ on $\curveA'$ got mapped to. 
Let $\eta$ denote the entire subcurve of $\curveB'$ before and including  $x$, and let $\gamma=\seq{\gamma_1,\ldots,\gamma_k}$ denote the entire subcurve after and including $x$. 
%
%
We wish to consider the curve $s\circ \gamma$, though it may contain redundant vertices. 
So define a curve $\curveB$, where $\curveB=\gamma$ if $x=s$. Now if $x\neq s$, note that $\gamma_1=x$ might lie on the segment $s\gamma_2$. (and note $\gamma_2,\ldots,\gamma_k$ are vertices of $\curveB'$). Thus in this case let $\curveB=\seq{s,\gamma_2,\ldots,\gamma_k}$ if $x$ lies on $s\gamma_2$, and otherwise $\curveB=s\circ \gamma$, as shown in \figref{fig:MV_s_cases_complex}.

%
Observe, that $\distFr{\curveA'}{\curveB}\leq \thresh$, since $p_1,p_2\in B(s,\delta)$ and so we can can stand still at $s$ on $\curveB$ while on $\curveA'$ we traverse $\seq{p_1,p_2,p_1,p_2,s}$, and then we can stand still at $s$ on $\curveA'$ while on $\curveB$ we traverse from $s$ to $x$, and then afterwards follow the portions of the $\thresh$-realizing traversal of $\curveA'$ and $\curveB'$ corresponding to their remaining subcurves. 
Moreover, the later part of the traversal just described also implies $\distFr{\curveA}{\curveB}\leq \thresh$.

Note that $|\curveB|\leq |\curveB'|$ in all cases except when $x$ is not on the segment $s\gamma_2$ ($s_3$ in \figref{fig:MV_s_cases_complex}), yet $x$ does lie on the segment $\curveB'_1\curveB'_2$ (in which case $\curveB'_2=\gamma_2$). However, we now argue that this is not possible, by arguing the $s$ must occur on $\eta$.

Suppose that $s$ does not occur on $\eta$. 
We claim that $x$ must occur strictly after $\curveB'_3$ on $\curveB'$, which would give a contradiction as then $|\curveB|<|\curveB'|$ (i.e.\ $\curveB'$ would not be $\mv{\curveA'}$ since we argued $\distFr{\curveA'}{\curveB}\leq \thresh$).
Let $\ell^\perp$ be the line perpendicular to $\ell$ and passing through $s$, and let  $H_1$ (resp.\ $H_2$) denote the open halfplane bounded by $\ell^\perp$ and containing $p_1$ (resp.\ $p_2$). Observe that the only point on $\ell^\perp$ within distance $\thresh$ to either $p_1$ or $p_2$ is $s$. Thus if $s$ does not occur on $\eta$, in order to match the subcurve $\seq{p_1,p_2,p_1,p_2}$, the curve $\eta$ must go from $H_1$ to $H_2$, back to $H_1$, and then back to $H_2$ again.
At minimum, this requires a vertex for the first visit to $H_1$, then a vertex later to turn from $H_2$ back to $H_1$, and then a third vertex to turn from $H_1$ back again to $H_2$.
Thus $x$ would occur strictly after $\curveB'_3$, giving a contradiction as described above, and so $s$ must occur on $\eta$.   



In summary, for the reduction we construct $\curveA'$, compute $\mv{\curveA'}=\curveB'$, and then compute a $\thresh$-traversal of $\curveA'$ and $\curveB'$ using the standard \Frechet distance algorithm. Then we construct the curve $\curveB$ as described above, except where we perform the above steps both with respect to the $s$ side and the $t$ side of the curve. 
As argued above, $\curveB$ is a minimum vertex curve starting at $s$ and ending at $t$ such that the \Frechet distance to $\curveA'$ is $\leq \thresh$.
Since we argued $\distFr{\curveA}{\curveB}\leq \thresh$, this implies it is a minimum vertex curve with \Frechet distance $\leq \thresh$ to $\curveA$ restricted to starting at $s$ and ending at $t$, 
because for any curve $\zeta$ starting at $s$ and ending at $t$, if $\distFr{\curveA}{\zeta}\leq \thresh$, then $\distFr{\curveA'}{\zeta}\leq \thresh$.

As, $\curveA'$ has only a constant number of vertices more than $\curveA$, computing $\mv{\curveA'}$ takes $O(m^2\log^2 m)$ time with \thmref{stabtime}, and this dominates the running time as computing $\curveB$ from $\mv{\curveA'}$ can be done in $O(m^2\log m)$ time with the standard \Frechet distance algorithm.



}



\end{document}